\documentclass[aps,prr,showpacs,superscriptaddress,twocolumn,raggedfooter,raggedbottom]{revtex4-2}

\usepackage{graphicx}
\usepackage{dcolumn}
\usepackage{bm}
\usepackage{amsmath,amssymb}
\usepackage{amsthm}
\usepackage[usenames]{color}
\usepackage{array}
\newtheorem{proposition}{Proposition}

\usepackage{soul}

\usepackage{hyperref}

\newcolumntype{M}[1]{>{\centering\arraybackslash}m{#1}}
\newcolumntype{N}{@{}m{0pt}@{}}

\begin{document}

\preprint{APS/123-QED}

\title{Time-of-arrival distributions for continuous quantum systems and application to quantum backflow}

\author{Mathieu Beau}
\affiliation{University of Massachusetts, Boston, MA, USA}

\author{Maximilien Barbier}
\affiliation{Scottish Universities Physics Alliance, University of the West of Scotland, Paisley PA1 2BE, Scotland, United Kingdom}

\author{Rafael Martellini}
\affiliation{Centre International de Valbonne, Sophia-Antipolis, France}

\author{Lionel Martellini}
\affiliation{EDHEC Business School, Nice, France}

\date{\today}

\begin{abstract}
Using standard results from statistics, we show that for any continuous quantum system (Gaussian or otherwise) and any observable $\widehat{A}$ (position or otherwise), the distribution $ \pi _{a}\left(t\right) $ of time measurement at a fixed state $a$ can be inferred from the distribution $ \rho _{t}\left( a\right) $ of a state measurement at a fixed time $t$ via the transformation $ \pi _{a}\left( t\right) \propto \left\vert \frac{\partial }{\partial t} \int_{-\infty }^{a}\rho _{t}\left( u\right) du \right\vert $. This finding suggests that the answer to the long-lasting time-of-arrival problem is in fact secretly hidden within the Born rule, and therefore does not require the introduction of a time operator or a commitment to a specific (e.g., Bohmian) ontology. The generality and versatility of the result are illustrated by applications to the time-of-arrival at a given location for a free particle in a superposed state and to the time required to reach a given velocity for a free-falling quantum particle. Our approach also offers a potentially promising new avenue toward the design of an experimental protocol for the yet-to-be-performed observation of the phenomenon of quantum backflow.
\end{abstract}

\maketitle

\pagebreak

\section{Introduction}\label{Section:Intro} 

Two main types of spatiotemporal measurements can be performed on a physical system. The first type involves measuring the position of the system, say a particle, at a particular point in time, while the second type involves measuring the time at which the particle passes a specific point in space.  An interesting contrast exists, however, between the status of these two types of measurements in the mathematical formulation of quantum mechanics: while the Born rule gives the probability density of a position measurement at a fixed time, there is no readily available rule for obtaining the probability density of a time measurement at a fixed position. That the standard quantum formalism remains silent about the predicted results of time-of-arrival measurements while such measurements are routinely performed in physical experiments is a rather disconcerting blind spot in our quantum mechanical description of the physical world. In a nutshell, the essence of what is sometimes known as the \textit{time-of-arrival (TOA) problem} is that no self-adjoint operator canonically conjugate to the Hamiltonian can be associated with time measurement \cite{pauli1933handbuch}. As a result, time is treated as a mere parameter in the quantum formalism. 

The quantum TOA problem has led to a long-lasting controversy in the literature, which is broadly divided into three main strands. The first approach to the TOA problem, often referred to as \textit{semiclassical approach}, or \textit{hybrid approach}, consists of using a classical equation of motion with uncertain quantum initial position and velocity conditions consistent with the position/momentum uncertainty relationship. This approach appears to be a natural and pragmatic first step towards the analysis of quantum TOA distributions and is particularly well-suited for the analysis of experimental results since it can accommodate the presence of various physical constraints in the practical implementation of the measurement process (see for example \cite{udovic1993neutron,kurtsiefer1995time,kothe2013time,dufour2014shaping,gliserin2016high,rousselle2022analysis}). While it can lead to reasonably good
approximations in the far field regime when detectors are placed far away
from the source, one key limit of this approach, however, is that it ignores the fully quantum mechanical information embedded in the dynamical propagation of the state of the system through the Schr\"{o}dinger
equation. The second approach typically referred to as the \textit{%
operator-based approach}, consists of quantizing the classical time of arrival to identify a suitable time operator by relaxing one of the constraints put forward in Pauli's argument \cite{pauli1933handbuch} against
the existence of a bona fide quantum time operator. This is achieved by either considering a non-self-adjoint operator, or an operator not canonically conjugate to the system Hamiltonian
(see for example \cite{aharonov1961time,giannitrapani1997positive,delgado1999quantum,galapon2004shouldn}). More generally, outcomes of quantum measurements can always be described by positive-operator-valued measures (POVMs), and various arrival time POVMs have been introduced in the literature (see \cite{vona2013does, TUMULKA2022168910} for a discussion).
Besides the intrinsically ad-hoc nature of the search for a quantum time operator or POVM, this approach does not easily lead to explicit results regarding the TOA distribution. For example \cite{Rovelli96} only deals with the free particle, for which
they can only obtain a semi-classical approximation. The free fall
problem is analyzed via a dedicated TOA operator in \cite{flores2019quantum}, but the authors also restrict the analysis to a semi-classical approximation and do not obtain any analytical expression for the TOA distribution or its moments (the associated eigenvalue problem is solved numerically by coarse-graining). Finally, the third approach to the TOA problem consists of addressing the question within Bohmian mechanics, which relies on the existence of physical trajectories. The key insight obtained with this approach is a remarkably simple expression for the distribution of the TOA in terms of the absolute value of the probability current subject to a suitable normalization (see for example \cite{leavens1993arrival,mckinnon1995distributions,daumer1997quantum,leavens1998,vona2013does}). This approach, however, involves a strong departure from the standard formalism, with an underlying ontology that is not subject to a wide consensus. 
 
The present paper sheds new light on this debate by showing that the answer to the long-lasting time-of-arrival problem is secretly hidden within the Born rule. More specifically, we use standard results from statistics, namely the so-called \textit{probability integral transform theorem} and the \textit{method of transformations}, to derive, for any (possibly non-Gaussian) continuous quantum system and any observable $\widehat{A}$ (possibly different from the position operator), the distribution $ \pi _{a}\left( t\right) $ of a time measurement at a fixed state $a$ from the distribution $ \rho _{t}\left( a\right) $ of a state measurement at a fixed time $t$. Specifically, we show that these two distributions are related by the following transformation (see Proposition 2 and the discussion that follows):
\begin{equation}
\pi _{a}\left( t\right) \propto \left\vert \frac{\partial }{\partial t} \int_{-\infty }^{a}\rho _{t}\left( u\right) du \right\vert. 
\end{equation}%

This analysis, which generalizes the Gaussian framework used in \cite{Beau24, beau2024mutual} to derive the time-of-arrival distribution for free-falling particles, can be used to relate the distribution of the time-of-arrival at a given position to the absolute value of the probability current at that position (see equation \eqref{pi_def2}), thus providing both a formal justification and a generalization for the result obtained with Bohmian mechanics for the specific case of the position operator. That we can confirm through statistical arguments the Bohmian distribution of the time-of-arrival is not surprising since Bohmian mechanics has been constructed to give the same statistical predictions as standard quantum mechanics if a measurement is performed \cite{bohm1952suggested}. Importantly, our results extend the analysis of time-of-arrival distributions to any observable whose spectrum is continuous: in contrast to the Bohmian approach, which explicitly relies on the notion of trajectories in physical space and has focused on time-of-arrival at a given position, our approach is general enough to apply not only to the time-of-arrival at a given position but also the time required for the particle to reach a given velocity or momentum, for example. While it should be emphasized that our results have been obtained without the need to invoke an underlying Bohmian ontology, our approach still involves some form of departure from the standard formalism in the sense that our distribution of the TOA has been obtained without \textit{a priori} relying on the notion of an associated POVM. 

On a different note, it should be noted that we are also able to obtain an explicit representation of time-of-arrival as a random variable (see equation (\ref{repres})), which can be used to derive approximate or exact analytical expressions for its various moments. This feature has been used in \cite{Beau24, beau2024mutual} to derive analytical expressions for the mean and standard deviation of the time-of-arrival of a free-falling particle at a given position in various regimes, results that cannot be obtained with the Bohmian approach which at best gives access to numerical estimates from the distribution function. 

A better analytical understanding of time-of-arrival distributions offers new perspectives for studying various quantum phenomena. Here we apply our results to propose a new experimental avenue regarding the yet-to-be-performed experimental observation of the phenomenon of quantum backflow. Broadly speaking, the latter refers to the counter-intuitive fact that a quantum particle can move in the direction opposite to its momentum. Interestingly, this effect was first identified in the context of quantum arrival times~\cite{Allcock69III}. A distinctive signature of the occurrence of quantum backflow is the change of sign of the probability current. Since our TOA distribution, when applied to the position operator, appears to be proportional to the absolute value of the current, quantum backflow can thus be seen to occur whenever our TOA distribution vanishes. This observation provides a simple new experimental scenario for observing this phenomenon by means of e.g. time-of-flight measurements. Our proposal hence extends the range of possible experimental schemes that could allow for the first experimental observation of the peculiar effect of quantum backflow.

The rest of the paper is organized as follows. In Section~\ref{gen_sec}, we present an
analysis of time measurements for continuous quantum systems and we
provide an explicit expression for the distribution of the TOA
of a generic observable at a given state. In Section~\ref{ex_sec}, we present examples of applications of the approach: in particular, we discuss the possible implications of our main result for the yet-to-be-performed experimental observation of the elusive phenomenon of quantum backflow. Finally, we present our
conclusions and suggestions for further research in Section~\ref{concl_sec}.

\section{General analysis of time-of-arrival for general quantum systems}\label{gen_sec}

We first present a stochastic representation for a general quantum system
and associated time measurements. We then introduce a set of mathematical
results and a general framework that can be used to derive the distribution
of the time-of-arrival at a given state for a generic observable.

Consider a quantum Hermitian operator $\widehat{A}$ and its eigenbasis $%
\left\{ \left\vert a\right\rangle \right\} $, where $a\in \mathbb{R}$ are the
eigenvalues of the operator. As already indicated, we solely focus in this paper on
operators admitting a continuous spectrum.
The spectral representation of the operator reads 
\begin{equation}
\widehat{A}=\int_{\mathbb{R}}da\ g(a)|a\rangle \langle a|,
\end{equation}%
where $g(a)$ is the spectral density of the operator. For example, the
operator position on a real line is given by $\widehat{x}=\int_{\mathbb{R}%
}dx\ x|x\rangle \langle x|$. 

Let $|\psi _{t}\rangle $ represent the state of the system at time $t\geq 0$%
. The dynamical evolution of the state is given by the Schr\"{o}dinger
equation 
\begin{equation*}
\widehat{H}(t)|\psi _{t}\rangle =i\hbar \frac{d}{dt}|\psi _{t}\rangle ,
\end{equation*}%
where $\widehat{H}(t)$ is the possibly time-dependent Hamiltonian for the
system. If we choose the basis $\{|a\rangle \}$, the wave function $\psi
_{t}(a)=\langle a|\psi _{t}\rangle $ satisfies the Schr\"{o}dinger equation: 
\begin{equation*}
\widehat{H}(t)\psi _{t}(a)=i\hbar \frac{\partial }{\partial t}\psi _{t}(a).
\end{equation*}

Consider next the random variable denoted by $A_{t}$ associated to the
measurement outcomes of the observable $\widehat{A}$ at time $t$. By the
Born rule, and assuming that the wave function $\psi _{t}(a)$ is square-integrable, we know that $A_{t}$ admits a time-dependent probability density function 
$\rho _{t}\left( a\right) $ given by $\rho _{t}\left( a\right) =\left\vert
\psi _{t}\left( a\right) \right\vert ^{2}$. Note that this stochastic representation of quantum measurement outcomes merely boils down to introducing some new notation; we
simply \textit{define} $A_{t}$ as representing the uncertain outcome of a 
\textit{first measurement} of the observable $\widehat{A}$ performed at time 
$t$, after the system has been prepared in the state represented by $\psi
_{0}$ and evolved according to the Schr\"{o}dinger equation\textit{,
with no prior measurement}, to the state $\psi _{t}$.

\subsection{A stochastic representation of quantum measurements for
continuous systems}

Using a standard result from probability theory known as the \textit{%
probability integral transform theorem}, sometimes also known as \textit{%
universality of the uniform} (see for example chapter 7 of \cite{shorack2000probability}), the following proposition shows that a
continuous time-dependent random variable $A_{t}$ can always be represented
as a time-dependent function of a \textit{time-independent} random variable $%
\xi $.

\begin{proposition}
Let $A_{t}$ be the continuous random variable associated with the
measurement outcome at time $t$ of an observable $\widehat{A}$ with
continuous eigenbasis $\left\{ \left\vert a\right\rangle \right\} $, and let 
$\rho _{t}\left( a\right) \equiv \left\vert \psi _{t}\left( a\right)
\right\vert ^{2}$ and $F_{t}\left( a\right) \equiv \int_{-\infty
}^{a}\rho _{t}\left( u\right) du$ be its time-dependent
probability density function (PDF) and time-dependent cumulative
distribution function (CDF), respectively. Let us further assume that $F_{t}$ is
invertible. Then the random variable $A_{t}$ can be written as 
\begin{equation}
A_{t}=F_{t}^{-1}\left( \xi \right) ,
\end{equation}%
where $\xi $ is uniformly distributed on $\left[ 0,1\right] $ and admits the
time-independent PDF 
\begin{equation}
f_{\xi }\left( y\right) =\left\{ 
\begin{array}{c}
1\text{ , if }y\in \left[ 0,1\right]  \\ 
0\text{ , if }y\notin \left[ 0,1\right] 
\end{array}%
\right. 
\end{equation}%
and time-independent CDF 
\begin{equation}
F_{\xi }\left( y\right) =\left\{ 
\begin{array}{l}
0\text{ , if }y<0 \\ 
y\text{ , if }y\in \left[ 0,1\right]  \\ 
1\text{ , if }y>1%
\end{array}%
\right. .
\end{equation}
\end{proposition}

\begin{proof}
Let us define the random variable $\xi \equiv F_{t}\left( A_{t}\right) $,
where $F_{t}\left( \cdot \right) $ is the CDF of $A_{t}$, and let $F_{\xi
}\left( \cdot \right) $ be the CDF of $\xi $. Note that $\xi $ takes values
in $\left[ 0,1\right] $. Assuming that $F_{t}$ is invertible, we have for any $y\in %
\left[ 0,1\right] $:%
\begin{eqnarray*}
F_{\xi }\left( y\right)  &=&\Pr \left( \xi \leq y\right) =\Pr \left(
F_{t}\left( A_{t}\right) \leq y\right)  \\
&=&\Pr \left( A_{t}\leq F_{t}^{-1}\left( y\right) \right) =F_{t}\left(
F_{t}^{-1}\left( y\right) \right) =y .
\end{eqnarray*}%
The CDF of $\xi $ is therefore $F_{\xi }\left( y\right) =y$, which shows
that $\xi $ has a uniform distribution on $\left[ 0,1\right] $, or $\xi
\hookrightarrow U\left( 0,1\right) $. This is the essence of the \textit{probability integral transform theorem}. Intuitively, this result
states that when $X$ is a random variable, the percentile of $X$ is
"uniform", e.g. the number of outcomes for $X$ falling into the 0\%-25\%
percentile should be the same as the number of outcomes for $X$ falling into
the 25\%-50\%. As a result, the CDF of $X$ is uniformly distributed. 
\end{proof}
While the bijective property needed to define the inverse function $F_{t}^{-1}$
a priori requires that the CDF $F_{t}$ is strictly monotonic with respect to time, we actually do not need to assume a strictly increasing CDF as long as the CDF is continuous. Indeed, one can repeat the proof with the quantile
function, which is essentially a generalized inverse function $%
F_{t}^{-1}\left( y\right) =\inf \left\{ x\text{ such that }F_{t}\left(
x\right) \geq y\right\} $ for $y\in \left[ 0,1\right] .$ With this
definition, we do have $F_{\xi }\left( y\right) =\Pr \left( \xi \leq
y\right) =\Pr \left( F_{t}\left( A_{t}\right) \leq y\right) =\Pr \left(
A_{t}\leq \inf \left\{ x\text{ such that }F_{t}\left( x\right) \geq
y\right\} \right) =\Pr \left( A_{t}\leq F_{t}^{-1}\left( y\right) \right)
=F_{t}\left( F_{t}^{-1}\left( y\right) \right) =y$.

In what follows, we show that the simple and general representation $A_{t}=F_{t}^{-1}\left( \xi \right) $
turns out to have important applications to the analysis of the time-of-arrival problems in quantum physics.
It is important to emphasize at this point that if the representation of $A_{t}$ as a time-dependent function
of a time-independent random variable $A_{t}=F_{t}^{-1}\left( \xi \right) $
with $\xi \hookrightarrow U\left( 0,1\right) $ is general enough to hold for any continuous
random variable $A_{t}$ with an invertible CDF, it is not necessarily
unique. For example if $A_{t}$ is normally distributed, that is if $%
A_{t}\hookrightarrow N\left( \mu _{A}\left( t\right) ,\sigma _{A}\left(
t\right) \right) $, then one may alternatively use the linear representation 
$A_{t}=\mu _{A}\left( t\right) +\xi ^{\prime }\sigma _{A}\left( t\right) $
where $\xi ^{\prime }\equiv \frac{A_{t}-\mu _{A}\left( t\right) }{\sigma
_{A}\left( t\right) }$ is a standardized Gaussian distribution $\xi ^{\prime
}\hookrightarrow N\left( 0,1\right) $. For Gaussian systems, this linear representation can be more convenient than
the nonlinear representation $A_{t}=F_{t}^{-1}\left( \xi \right) $%
\textbf{\ }with $\xi \hookrightarrow U\left( 0,1\right) $. Note that both $%
\xi $ and $\xi ^{\prime }$ are time-independent random variables when $A_{t}$
is a continuous Gaussian random variable. Actually, we show below (see proof of Proposition 3) that both
representations lead to the same distribution for the TOA of the
observable $\widehat{A}$ when $A_{t}$ is Gaussian (see Proposition 3). On the other hand, when $%
A_{t}$ is non-Gaussian, $\xi \equiv F_{t}\left( A_{t}\right) $ remains a
time-independent random variable, but $\xi ^{\prime }\equiv \frac{A_{t}-\mu
_{A}\left( t\right) }{\sigma _{A}\left( t\right) }$ will exhibit
time-dependencies in its distribution and cumulative distribution functions.
The representation $A_{t}=F_{t}^{-1}\left( \xi \right) $ is therefore more
general and holds for all systems while the linear representation $A_{t}=\mu
_{A}\left( t\right) +\xi ^{\prime }\sigma _{A}\left( t\right) $ only works
for Gaussian systems.

We now discuss the application of this stochastic representation to the analysis of TOA.

\subsection{Distribution of time-of-arrival measurements for continuous
systems}

The notion of time-of-arrival we consider in this paper is extremely general. For example, it can be the time-of-arrival at a given position, but also the time required for a system to reach a given velocity or momentum, for example. Formally, let us first define the random time-of-arrival (TOA) $T_{a}$ of the
observable $\widehat{A}$ at a measured state $a$ starting from some
initial state. What should be regarded as a \textit{time-of-first-measurement} at state $a$ is thus represented by the following random variable: 
\begin{equation}
T_{a}\equiv \inf \left\{ t\text{ such that }A_{t}=a\right\} .  \label{ttdef}
\end{equation}

Just as $A_{t}$ is a random variable that represents the measurement
outcome of the observable $\widehat{A}$ at time $t$, $T_{a}$ is a random
variable that represents the time until the first measurement of the
observable $\widehat{A}$ yields the value $A_{t}=a$. This duality is reflected in the symmetry in notation between $%
T_{a}$, which represents a measured time of arrival at a given state $a$, and $A_{t}$, which represents a measured state
at a given time $t$.

To clarify without any ambiguity what
we mean by the time of a first measurement at state $a$, we outline that
this time measurement corresponds to the following stylized experimental
procedure: (i) we place a single detector which can only detect the state $%
\left\vert a\right\rangle $; (ii) we prepare the system at time $0$ in some
initial state ($a^{\prime }$) and we turn the detector on at some time $t$;
(iii) we record $1$ if the measured value equals $a$ for this particular
time $t$ and $0$ otherwise; (iv) we then repeat the steps (i)-(iii) $N$
times while keeping the exact same time $t$; (v) we finally repeat the
steps (i)-(iv) by letting $t$ vary, with a small enough temporal resolution $%
\delta t$ (hence, $t_{0}=0,t_{1}=\delta t,$\textperiodcentered
\textperiodcentered \textperiodcentered $,t_{k}=k\delta t,$%
\textperiodcentered \textperiodcentered \textperiodcentered $,t_{n}=n\delta t
$). This procedure allows us to reconstruct the whole distribution $\pi
_{a}(t)$ of the random variable $T_{a}$, which again can be regarded as a
stochastic time of arrival to the fixed eigenvalue $a$ (see in \cite{Beau24, beau2024mutual} for a more detailed description of the experimental protocol that can be used to obtain the distribution of the TOA at a given position).

The following proposition gives a simple explicit expression for the
probability density function for the TOA $T_{a}$. This result is extremely general and allows us to obtain the
TOA distribution for any continuous system, Gaussian or non-Gaussian. 

\begin{proposition}
Let $T_{a}\equiv \inf \left\{ t\text{ such that }A_{t}=a\right\} $ be the
random time until a first measurement yields the outcome $A_{t}=a$ for some
state $a$. Denoting by $\rho _{t}\left( a\right) \equiv \left\vert \psi
_{t}\left( a\right) \right\vert ^{2}$ and $F_{t}\left( a\right) \equiv
\int_{-\infty }^{a}\rho _{t}\left( u\right) du$ the PDF and CDF
of $A_{t}$, respectively, the probability
distribution function of the random variable $T_{a}$, denoted by $\pi
_{a}\left( t\right) $, is related to the CDF $F_{t}\left( a\right)$ by the following transformation:%
\begin{equation}
\pi _{a}\left( t\right) \propto \left\vert \frac{\partial }{\partial t}F_{t}\left(
a\right) \right\vert .
\label{pi_def}
\end{equation}%
\end{proposition}

The presence, or the absence, of a normalization factor in equation (\ref{pi_def}) can be analyzed as follows. First note that for all systems such that there is a positive probability $p$ that the particle never reaches a detector located at position $x$ (more generally never reaches the measured state $a$), the unconditional distribution of the time-of-arrival at $x$ (more generally at $a$) should not integrate to $1$ but to $1-p$. As a matter of fact, a straightforward application of Bayes' theorem suggests that the normalized version of the function $\pi _{a}\left( t\right) $ defines the \textit{conditional} distribution given that a detection actually occurs, that is the distribution for those particles that are actually measured in state $a$ as captured by $T_{a} \geq 0$:
\begin{multline*}
\Pr \left( \left. T_{a}\in \left[ t,t+dt\right] \right| T_{a}\geq 0\right) =\\ 
\dfrac{\Pr \left( \left. T_{a}\geq 0\right| T_{a}\in \left[ t,t+dt\right]
\right) \Pr \left( T_{a}\in \left[ t,t+dt\right] \right) }{\Pr \left(
T_{a}\geq 0\right) } .
\end{multline*}
Further noticing that $\Pr \left( \left. T_{a}\geq 0\right| T_{a}\in \left[ t,t+dt%
\right] \right) =1$, $\Pr \left( T_{a}\in \left[ t,t+dt\right] \right) =\pi
_{a}\left( t\right)dt $ and also that $\Pr \left( T_{a}\geq 0\right) =\int_{0}^{\infty
}\pi _{a}\left( s\right) ds$, we thus finally obtain
\begin{equation}%
\Pr \left( \left. T_{a}\in %
\left[ t,t+dt\right] \right| T_{a}\geq 0\right) =\dfrac{\pi _{a}\left(
t\right)dt }{\int_{0}^{\infty }\pi _{a}\left( s\right) ds}.
\end{equation}%

In this context, it is a matter of choice whether or not one should normalize the expression for $\pi _{a}\left( t\right)$ and it is only when the
condition $T_{a}\geq 0$ is satisfied for all values for $\xi $ that the normalized and non-normalized versions of the distribution coincide. 

\begin{proof}
While the proof is presented for simplicity of exposure under the restrictive assumption that the function $F_{t}\left( a\right)$ is strictly monotonic with respect to the time variable $t$, it should be noted that this assumption is not required for our result to hold. Indeed if $F_{t}\left( a\right) $ is not strictly monotonic with respect to $t$, one would end up with multiple solutions for the expression of $T_a$ as a function of $\xi$ but all of these functions relating $T_a$ to $\xi$ would have the same inverse function, given by $\xi =F_{T_{a}}\left( a\right)$, and this is the only ingredient that is needed in the application of the method of transformations. By the definition of $T_{a}$ in (\ref{ttdef}), we have $A_{T_{a}}=a$ almost
surely. Assuming that the function $F_{t}\left( a\right)$ is strictly monotonic, the inverse function $F_{t}^{-1}\left( a\right) $ exists
and is unique, and we obtain from the representation result $A_{t}=F_{t}^{-1}\left( \xi \right) $
in Proposition 1: 
\begin{equation*}
A_{T_{a}}=a\Longleftrightarrow F_{T_{a}}^{-1}\left( \xi \right) =a
\end{equation*}%
or 
\begin{equation*}
\xi =F_{T_{a}}\left( a\right) \equiv h_{a}^{-1}\left( T_{a}\right) ,
\end{equation*}%
for some function $h_{a}\left( \cdot \right) $ such that 
\begin{equation*}
h_{a}^{-1}\left( t\right) =F_{t}\left( a\right) .
\end{equation*}%
We thus obtain: 
\begin{equation*}
\xi =h_{a}^{-1}\left( T_{a}\right) \Longrightarrow T_{a}=h_{a}\left( \xi
\right) ,
\end{equation*}%
subject to the condition $T_{a} \geq 0$. When this condition is satisfied, we can thus use the representation for $A_{t}$, $%
A_{t}=F_{t}^{-1}\left( \xi \right) $, to obtain a representation for $T_{a}$%
, namely%
\begin{equation}
T_{a}=h_{a}\left( \xi \right) .  \label{repres}
\end{equation}%
This representation $T_{a}=h_{a}\left( \xi \right) $ is convenient because
it directly allows us to obtain the PDF $\pi _{a}\left( t\right) $ for $T_{a}
$ as a transformation of the PDF $f_{\xi }$ of $\xi $ by using a standard
result in probability sometimes called the "method of transformations" (see for
example theorem 4.1 in chapter 4.1.3 of \cite{StochTextbook}):%
\begin{equation}
\pi _{a}\left( t\right) =f_{\xi }\left( h_{a}^{-1}\left( t\right) \right)
\times \left\vert \frac{\partial }{\partial t}h_{a}^{-1}\left( t\right)
\right\vert \text{ } . \label{pit}
\end{equation}%
Using $h_{a}^{-1}\left( t\right) =F_{t}\left( a\right) $, equation \eqref{pit}
becomes%
\begin{equation*}
\pi _{a}\left( t\right) =f_{\xi }\left( F_{t}\left( a\right) \right) \times
\left\vert \frac{\partial }{\partial t}F_{t}\left( a\right) \right\vert 
\text{ }
\end{equation*}%
or simply%
\begin{equation*}
\pi _{a}\left( t\right) =\left\vert \frac{\partial }{\partial t}F_{t}\left(
a\right) \right\vert \text{ }
\end{equation*}%
since $f_{\xi }\left( y\right) =1$ if $y\in \left[ 0,1\right] $ when $\xi
\hookrightarrow U\left( 0,1\right) $. 
\end{proof}

It is important to note that the expression in equation \eqref{pi_def} can be related to the probability current for systems where this quantity is well-defined. For instance, if one considers a quantum particle moving in one-dimensional space with a Hamiltonian $\hat{H} = -\frac{\hbar^2}{2m}\frac{\partial^2}{\partial x^2}+V(x,t)$, then the expression of the current is 
\begin{equation}\label{Eq:CurrentDef}
j_{t}( x ) \equiv \frac{\hbar }{2mi}\left[ \psi ^{\ast
}(t,x)\frac{\partial }{\partial x}\psi (t,x)-\psi (t,x)\frac{\partial }{%
\partial x}\psi ^{\ast }(t,x)\right] 
\end{equation}

In this case, the continuity equation $%
\frac{\partial }{\partial t}\rho _{t}\left( x\right) +\frac{\partial }{%
\partial x}j_{t}\left( x\right) =0$ 
implies that
\begin{equation*}
\frac{\partial }{\partial t}F_{t}\left( x\right) =\int_{-\infty
}^{x}\frac{\partial }{\partial t}\rho _{t}\left( u\right) du=-j_{t}\left(x\right) , 
\end{equation*}
assuming that $j_t(x)\rightarrow 0$ when $x\rightarrow -\infty$.
Therefore, we have the following expression of the TOA distribution of the particle at a given position $x$:
\begin{equation}\label{pi_def2}
\pi _{x}\left( t\right) \propto \left\vert
j_{t}\left( x \right) \right\vert.
\end{equation}
Notice that this relation is valid for any expression of the current, as long as the continuity equation holds true and that the current vanishes when $x\rightarrow-\infty$ (if the latter condition is not valid, the expression \eqref{pi_def2} can be easily modified accordingly). For example, one can find a more general expression of the current for a spin particle in a magnetic field in~\cite{LandauQM77}. 

That we can formally relate the distribution of the time-of-arrival at a given position to the absolute value of the probability current at that position provides an independent justification for the result obtained with Bohmian mechanics without having to adhere to the Bohmian formulation of quantum mechanics. (Note that the justification for the presence of the absolute value directly follows from the method of transformations.) In fact, proposition 2 generalizes the Bohmian prediction in two important directions: first, it provides a general expression for TOA distributions even for systems where the probability current is ill-defined, and secondly it extends to any observable the Bohmian distribution that is only derived for the position operator.

Before
we turn to examples of applications in the next section, let us remark
that we can calculate as follows the mean and standard deviation of the TOA $T_{a}$ from its distribution function $\pi _{a}$ (possibly normalized) given in
Proposition 2: 
\begin{subequations}
\begin{equation}
   \left\langle T_{a}\right\rangle  =\int_{0}^{+\infty }t\pi _{a}\left( t\right) dt,  \label{mean1}  
\end{equation}
\begin{eqnarray}
    \Delta T_{a} =\sqrt{\int_{0}^{+%
\infty }t^{2}\pi _{a}\left( t\right) dt-\left( \int_{0}^{+\infty
}t\pi _{a}\left( t\right) dt\right) ^{2}}.  \label{var1}
\end{eqnarray}
\end{subequations}
Alternatively, these quantities can be obtained directly, and sometimes more conveniently (see \cite{Beau24, beau2024mutual}), by taking the
moments of the representation (\ref{repres}) when such a
representation is available analytically: \begin{subequations}
\begin{equation}
   \left\langle T_{a}\right\rangle  =\mathbb{E}\left[ h_{a}\left( \xi \right) %
\right] ,  \label{mean2}
\end{equation}
\begin{eqnarray}
    \Delta T_{a} =\sqrt{\mathbb{E}\left[ h_{a}^{2}\left( \xi \right) \right]
-\left( \mathbb{E}\left[ h_{a}\left( \xi \right) \right] \right) ^{2}} .
\label{var2}
\end{eqnarray}
\end{subequations}

\subsection{Distribution of time of arrival measurements for Gaussian systems}

We now specialize the analysis to Gaussian systems, and we show that our general approach nests the Gaussian framework used in \cite{Beau24, beau2024mutual} as a special case.

\begin{proposition}
Let us consider an observable $\widehat{A}$ with a Gaussian probability
distribution function: 
\begin{equation}
\rho _{t}\left( a\right) =\frac{1}{\sqrt{2\pi }\sigma _{A}(t)}e^{-\frac{%
(a-\mu _{A}(t))^{2}}{2\sigma _{A}(t)^{2}}},
\end{equation}%
where $\mu _{A}(t)$ is the mean value, $\sigma _{A}(t)$ is the standard deviation, and
let $T_{a}\equiv \inf \left\{ t\text{ such that }A_{t}=a\right\} $ be the
random time until a first measurement yields the outcome $A_{t}=a$ for some
state $a$. The probability distribution function of $T_{a}$ is given by: 
\begin{equation}\label{Eq:pi:Gaussian}
\pi _{a}\left( t\right) =\frac{1}{\sqrt{2\pi }}e^{-\frac{(u-\mu _{A}(t))^{2}%
}{2\sigma _{A}(t)^{2}}}\frac{\partial }{\partial t}\left( \frac{a-\mu _{A}(t)%
}{\sigma _{A}(t)}\right) .
\end{equation}
\end{proposition}

\begin{proof}
From Proposition 2, we obtain that the probability distribution function of
the TOA $T_{a}\equiv \inf \left\{ t\text{ such that }%
A_{t}=a\right\} $ is given by 
\begin{equation*}
\pi _{a}\left( t\right) =\left\vert \frac{\partial }{\partial t}F_{t}\left(
a\right) \right\vert ,
\end{equation*}%
where 
\begin{align*}
F_{t}\left( a\right) &= 
\int_{-\infty }^{a}\frac{1}{\sqrt{2\pi }}e^{-%
\frac{(u-\mu _{A}(t))^{2}}{2\sigma _{A}(t)^{2}}}du \\
&=\frac{1}{2}\left( 1+\text{erf}\left( \frac{a-\mu _{A}(t)}{\sigma _{A}(t)\sqrt{2}}\right) \right) ,
\end{align*}
with $\text{erf}\left( z\right) = \frac{2}{\sqrt{\pi }}\int_{-\infty }^{z}e^{-u^2}du$. 
Noting that $\text{erf}^{\prime }\left( z\right) = \frac{2}{\sqrt{\pi}}e^{-z^2}$, we have
\begin{eqnarray*}
\pi _{a}\left( t\right)  &=&\left|\frac{\partial }{\partial t}\left[ \frac{1}{2}%
\left( 1+\text{erf}\left( \frac{a-\mu _{A}(t)}{\sigma _{A}(t)\sqrt{2}}%
\right) \right) \right]\right|  \\
&=&\left|\frac{\partial }{\partial t}\left( \frac{a-\mu _{A}(t)}{\sigma _{A}(t)%
\sqrt{2}}\right)\right| \frac{1}{2}\left( 1+\text{erf}^{\prime }\left( \frac{%
a-\mu _{A}(t)}{\sigma _{A}(t)\sqrt{2}}\right) \right)  \\
&=&\left|\frac{\partial }{\partial t}\left( \frac{a-\mu _{A}(t)}{\sigma _{A}(t)}%
\right)\right| \frac{1}{\sqrt{2\pi }}e^{-\frac{(u-\mu _{A}(t))^{2}}{2\sigma
_{A}(t)^{2}}} .
\end{eqnarray*}%
As discussed in a previous remark, the same result can actually be obtained via a
linear representation of $A_{t}$ in terms of a standardized Gaussian
variable.  To see this, let $\xi ^{\prime }\equiv \frac{A_{t}-\mu _{A}\left(
t\right) }{\sigma _{A}\left( t\right) }\hookrightarrow N\left( 0,1\right) $
and let $f_{\xi ^{\prime }}(a)=\frac{1}{\sqrt{2\pi }}e^{-\frac{(u-\mu
_{A}(t))^{2}}{2\sigma _{A}(t)^{2}}}$ be its PDF. This alternative stochastic
representation can be used to find the expression of $T_{a}$ by solving the
equation $a=A_{T_{a}}$: 
\begin{equation*}
a=\mu _{A}(T_{a})+\xi ^{\prime }\sigma _{A}(T_{a}),\ 
\end{equation*}%
which implies%
\begin{equation*}
\xi ^{\prime }=h_{a}^{-1}(T_{a})=\frac{a-\mu _{A}(T_{a})}{\sigma _{A}(T_{a})}%
,\ 
\end{equation*}%
where $h_{a}^{-1}(t)=\frac{a-\mu _{A}(t)}{\sigma _{A}(t)}$ is an invertible
function that can be solved for a given system. Assuming again that the
function $h_{a}$ is strictly monotonic, the ``method of transformation" can
be applied again to give the probability distribution $\pi _{a}\left(
t\right) $ for the transition time $T_{a}$ at the detection value $a$ as: 
\begin{equation*}
\pi _{a}\left( t\right) =f_{\xi ^{\prime }}(h_{a}^{-1}(t))\times \left\vert 
\frac{\partial }{\partial t}h_{a}^{-1}(t)\right\vert ,
\end{equation*}%
and we indeed recover%
\begin{equation*}
\pi _{a}\left( t\right) =\frac{1}{\sqrt{2\pi }}e^{-\frac{(a-\mu _{A}(t))^{2}%
}{2\sigma _{A}(t)^{2}}}\times \left|\frac{\partial }{\partial t}\left( \frac{a-\mu
_{A}(t)}{\sigma _{A}(t)}\right)\right| .
\end{equation*}
This approach has been used in \cite{Beau24} in the context of an application to the
distribution of the time-of-arrival at a given position for a free-falling quantum particle.
\end{proof}

The result in Proposition 3 is in
fact of wide application given that Gaussian systems are ubiquitous in
quantum physics. In practice, the expression for $\pi _{a}\left(
t\right) $ can be computed based on the specific form of the functions $\mu
_{A}(t)$ and $\sigma _{A}(t)$ for the system under analysis, and the mean
and standard-deviation of $T_{a}$ can be obtained from equations (\ref{mean1}%
) and (\ref{var1}) or (\ref{mean2}) and (\ref{var2}).

\section{Examples of application}\label{ex_sec}

In this section, we present three examples of applications of the approach. We first calculate the time-of-arrival at a given velocity for a Gaussian free-falling particle. Then, we discuss the possible implications of our main result regarding the yet-to-be-performed observation of the elusive phenomenon of quantum backflow. Finally, we study the time-of-arrival at a given location for a superposition of two Gaussian wave packets. 

\subsection{Time of arrival at a given velocity for the free-falling particle}\label{TOA_free_subsec}

In \cite{Beau24, beau2024mutual} the linear stochastic representation $A_{t}=\mu _{A}\left( t\right) +\xi ^{\prime }\sigma _{A}\left( t\right) $
where $\xi ^{\prime }\equiv \frac{A_{t}-\mu _{A}\left( t\right) }{\sigma
_{A}\left( t\right) }$ is a standardized Gaussian distribution $\xi ^{\prime
}\hookrightarrow N\left( 0,1\right) $ has been used to derive the
distribution of the time-of-arrival $T_{x}$ at a given position $X_{t}=x$ for a free falling quantum
particle, a system for which $\mu
_{X}(t)=x_{0}+v_{0}t+\frac{1}{2}gt^{2}$ and $\sigma _{A}(t)=\sigma \sqrt{1+%
\frac{t^{2}\hbar ^{2}}{4m^{2}\sigma ^{4}}}$, where $\sigma$ is the initial spread of the Gaussian wave packet, $m$ is the mass of the particle, and $g$ is the acceleration of gravity. In what follows, we turn to a
slightly different problem, namely the problem of estimating the time $T_{v}$
required for the free-falling particle to reach a certain velocity $v$,
using again our stochastic representation methodology for obtaining the
density distribution of $T_{v}$. To do this, we first need to switch bases and work with
the momentum operator $\widehat{p}=\int_{\mathbb{R}}dp\ p|p\rangle \langle p|$, where the position operator has the following representation $\widehat{x} = i\hbar\frac{\partial}{\partial p}$
. The Schr\"{o}dinger equation satisfied by the wave function in the
momentum basis is the following: 
\begin{equation}\label{Eq:Schro:freefall:momentum}
\left( \frac{p^{2}}{2m}-i mg\hbar\frac{\partial }{\partial p}\right) \psi(p,t)=i\hbar \frac{\partial }{\partial t}\psi(p,t),
\end{equation}
where the initial condition is given by 
\begin{equation*}
\psi(p,0)=\frac{1}{(2\pi \sigma_p^2)^{1/4}}e^{-\frac{(p-p_{0})^{2}}{4\sigma_p^2}} ,
\end{equation*}
where the mean value of the momentum operator is $\langle \widehat{p}\rangle
=p_{0}$ and where its standard deviation is $\sigma_p = \frac{\hbar }{2\sigma }$. The solution to the Schr\"{o}dinger equation is 
\begin{equation*}
\psi(p,t) = \frac{1}{(2\pi \sigma_p^2)^{1/4}}  e^{-i\frac{p^{3}}{6m^2g\hbar}} e^{i\frac{(p-mgt)^{3}}{6m^2g\hbar}} e^{-\frac{(p-p_c(t))^2}{4\sigma_p^2}} ,
\end{equation*}
where $p_{c}(t)=mgt+mv_{0}$ is the classical momentum. 
Hence, we find that the distribution of the momentum is Gaussian: 
\begin{equation*}
\rho _{t}(p)=\vert \psi(p,t)\vert^2=\frac{1}{\sqrt{2\pi \sigma _{p}^2}}e^{-\frac{(p-p_{c}(t))^2}{%
2\sigma _{p}}} .
\end{equation*}
As a result, we obtain the following stochastic representation for the momentum
measurement $P_{t}$ at time $t$:
\begin{equation}\label{Eq:freefall:momentum:rva}
P_{t}=p_{c}(t)+\xi \sigma _{p},
\end{equation}
where $\xi \hookrightarrow \mathcal{N}(0,1)$. To find the time $T_{v}$ of
arrival at a given velocity $v=p/m$, we must solve the
equation: 
$
P_{T_{v}}=mv .
$
Taking $v_{0}=0$, we find 
$
mv=mgT_{v}+\xi \sigma _{p},
$
implying that $T_{v}$ is finally given by 
\begin{equation*}
T_{v}=\frac{v-\xi \sigma _{v}}{g},
\end{equation*}
where $\sigma_{v}=\frac{\sigma_{P}}{m}=\frac{\hbar }{2m\sigma }$. Hence, we obtain
a simple linear relation between $\xi $ and $T_{v}$. It follows that $T_{v}$
is Gaussian, with a mean value and a standard deviation given by: 
\begin{subequations}
\begin{equation}\label{Eq:Tv:mean}
\ \left\langle T_{v}\right\rangle = t_{c}=\frac{v}{g} ,
\end{equation}
\begin{equation}\label{Eq:Tv:std}
\Delta T_{v} =\frac{\sigma _{v}}{g}=\frac{\hbar }{2mg\sigma } .
\end{equation}%
\end{subequations}
We thus obtain that the mean TOA for the velocity coincides with the
classical value. This is in contrast with the TOA at a given position $T_x$, where the mean is strictly greater than the classical value due to the presence of quantum corrections (see \cite{Beau24,beau2024mutual}). We also find that the standard deviation of the TOA at a given velocity is proportional to the standard
deviation of the velocity, with a constant of proportionality $1/g$. 
Note that we can confirm from equation \eqref{Eq:pi:Gaussian} that the time distribution $\pi_p(t)$ of a free-falling particle to reach the momentum $p$ indeed admits the following Gaussian density:
$$
    \pi_p(t) = mg\rho_{t}(p) = \frac{1}{\sqrt{2\pi\tau^2}}e^{-\frac{(t-t_p)^2}{2\tau^2}} ,
$$
where $t_p = \frac{p-mv_0}{mg}$ is the mean value of the distribution (consistently with equation \eqref{Eq:Tv:mean}) and $\tau=\frac{\sigma_p}{mg}=\frac{\hbar}{2mg\sigma}$ is the standard deviation of the distribution (consistently with equation \eqref{Eq:Tv:std}). \\

\subsection{TOA distribution and quantum backflow}\label{QB_subsec}

Here we apply the above description of arrival times to the peculiar phenomenon of quantum backflow. In the simplest scenario of a free nonrelativistic quantum particle moving in one dimension along the $x$ axis, quantum backflow pertains to the fact that the probability current $j_t(x)$ can be negative even if the momentum of the particle is positive with probability $1$. Historically, quantum backflow is intimately connected to the question of arrival times in quantum mechanics: indeed, this effect was introduced in 1969 as the origin of nonclassical contributions to the time-of-arrival probability distribution proposed by Allcock~\cite{Allcock69III}.

The first systematic study of quantum backflow was then performed by Bracken and Melloy in 1994~\cite{Bracken94}. Various aspects of this intriguing phenomenon have been studied within the last three decades, including free-fall~\cite{Melloy98_1}, relativistic~\cite{Melloy98_1,Melloy98_2, Ashfaque19}, many-particle~\cite{Barbier20} or two-dimensional~\cite{Strange12,Barbier23} scenarios. Quantum backflow is fundamentally rooted in the principle of superposition and hence provides yet another peculiar manifestation of interference. As such, effects akin to quantum backflow can arise for other kinds of waves. For instance, backflow has also been discussed for light, i.e. electromagnetic waves described by Maxwell's equations~\cite{Berry10}. Whereas optical backflow has recently been observed~\cite{Eliezer20,Daniel22}, an experimental observation of quantum backflow remains to be performed.

To this respect, proposals based on Bose-Einstein condensates have for instance been put forward~\cite{Palmero13, Mardonov14}. An alternative, ``experiment-friendly'' formulation of quantum backflow has also been recently proposed~\cite{Miller21}; however, the latter has been shown~\cite{Barbier21} not to be equivalent to the original formulation of quantum backflow.

In principle, a measurement of the wave function (by means of experimental techniques such as the ones discussed in~\cite{Lundeen11, Pan19, Zhang19, Sahoo20}) could incidentally allow to experimentally construct the probability current itself. Interestingly, our construction of the TOA distribution~\eqref{pi_def} can offer an alternative route: instead of focusing on the probability current, we propose to instead measure times of arrival (or times of flight), which are much more accessible in practice. The simplicity of the underlying experimental scheme hence makes our TOA distribution a promising candidate for a prospective experimental observation of the elusive effect of quantum backflow.

To illustrate how our TOA distribution~\eqref{pi_def} can be linked to quantum backflow, we consider a free quantum particle moving in one dimension along the $x$ axis with a positive momentum. Quantum backflow then occurs if the current $j_t(x)$ becomes negative at some space-time point $(x,t)$. In particular, a signature of the occurrence of quantum backflow is thus the change of sign of the current $j_t(x)$. We can now combine this with the fact that the TOA distribution~\eqref{pi_def} is valid for an arbitrary quantum state, and thus in particular also for a backflowing state. Remarkably, a signature of the occurrence of quantum backflow can then simply be taken to coincide with the vanishing of $\pi_{x}\left( t\right)$. This signature can thus be identified experimentally by means of measurements of arrival times (or times of flight). Again, one advantage of such a measurement scheme is that it would not require to measure the probability current itself. Indeed, the TOA distribution $\pi_{x}\left( t\right)$ is constructed by measuring the times at which a detector located at position $x$ detects the particle~\cite{Beau24}.

We illustrate this general idea with an explicit example. We consider a particle moving freely along the $x$ axis, and choose the state that was considered by Bracken and Melloy in~\cite{Bracken94}, for which the initial momentum wave function $\phi(p)$ at time $t=0$ is given by
\begin{equation}
    \phi(p) = \left\{\begin{array}{ll}
0 \qquad &, \quad \text{if} \quad p < 0 \\[0.4cm]
\frac{18}{\sqrt{35 \alpha^3}} \, p \left( e^{-p/\alpha} - \frac{1}{6} e^{-p/2 \alpha} \right) \qquad &, \quad \text{if} \quad p > 0
\end{array}\right. \, ,
\label{phi_p_QB}
\end{equation}
where $\alpha > 0$ is a constant with the dimension of momentum. The rationale for considering the particular state~\eqref{phi_p_QB} is twofold: (i) First, it allows to simplify our present discussion since we can readily observe numerically that the probability current at position $x=0$ only changes sign once with respect to time (see Fig.~\ref{TOA_QB_fig} below). (ii) Second, it is of historical relevance since it is the first normalized state for which quantum backflow has been explicitly shown to occur~\cite{Bracken94}. For this reason, we will refer to this particular state as the Bracken-Melloy state.

The fact that $\phi(p) = 0$ if $p<0$ expresses the fact that the momentum of the particle at time $t=0$ is, with certainty, positive. In other words, a measurement of the momentum of the particle at time $t=0$ will necessarily return a number $p$ that is positive (though the actual positive measured value of $p$ is uncertain). Because the particle is free by assumption, this property remains true at any later time $t>0$ since the position wave function $\psi(x,t)$ is then given by
\begin{equation}
    \psi(x,t)  = \frac{1}{\sqrt{2 \pi \hbar}} \int_{0}^{\infty} dp \, e^{- i p^2 t / 2 m \hbar} e^{ixp / \hbar} \phi(p) \,.
    \label{psi_x_t_Fourier}
\end{equation}
Substituting~\eqref{phi_p_QB} into~\eqref{psi_x_t_Fourier} and introducing the dimensionless variables
\begin{equation}
    x' \equiv \frac{\alpha x}{\hbar} \qquad \text{and} \qquad t' \equiv \frac{\alpha^2 t}{m \hbar} \,,
    \label{dimensionless_x_t}
\end{equation}
one can show~\cite{Bracken94, Barbier20} that the Bracken-Melloy state $\psi(x,t)$ can be written as
\begin{equation*}
    \psi(x,t) = \sqrt{\frac{\alpha}{\hbar}} \, \Psi(x',t') \,,
\end{equation*}
where $\Psi$ is a dimensionless function given by
\begin{widetext}
\begin{align}
    \Psi(x',t') = &- \frac{18}{\sqrt{70 \pi}} \left( \frac{5i}{6t'} + \sqrt{\frac{\pi}{4t'^3}} (i-1) \left\{ (x'+i) \exp \left[ \frac{i}{2 t'} (x'+i)^2 \right] \mathrm{erfc} \left[ - \frac{(1+i)(x'+i)}{\sqrt{4 t'}} \right] \right. \right. \nonumber\\[0.5cm]
    &- \left. \left. \frac{2x'+i}{12} \exp \left[ \frac{i}{8 t'} (2 x'+i)^2 \right] \mathrm{erfc} \left[ - \frac{(1+i)(2x'+i)}{\sqrt{16 t'}} \right] \right\} \right) \label{QB_state} \,,
\end{align}
\end{widetext}

\noindent where $\mathrm{erfc} \left( z \right)$ denotes the complementary error function, which is defined by
\begin{equation*}
    \mathrm{erfc} \left( z \right) = 1 - \mathrm{erf} \left( z \right) = \frac{2}{\sqrt{\pi}} \int_z^{\infty} dy \, e^{-y^2} \, .
\end{equation*}

We now consider the probability current $j(x,t)$ that corresponds to the Bracken-Melloy state~\eqref{QB_state}, which in view of~\eqref{dimensionless_x_t} can be expressed as
\begin{equation}
    j(x,t) = \frac{\alpha^2}{m \hbar} \, \mathcal{J} \left(x', t' \right)
    \label{current_rel}
\end{equation}
in terms of the dimensionless current $\mathcal{J}$ given by
\begin{equation*}
    \mathcal{J} \left(x', t' \right) = - \frac{i}{2} \left[ \Psi^* \frac{\partial \Psi}{\partial x'} - \Psi \frac{\partial \Psi^*}{\partial x'} \right] \,.
\end{equation*}
Focusing on the position $x=0$, Fig.~\ref{TOA_QB_fig} shows the (dimensionless) current $\mathcal{J} \left(0, t' \right)$ (dash-dotted orange line) as a function of the (dimensionless) time $t'$. We can readily see that $\mathcal{J} \left(0, t' \right)$ vanishes at $t' = t_0^{\prime}$, which we can numerically estimate to be $t_0^{\prime} \approx 0.021$. From ~\eqref{current_rel}, this means that $j(0,t)$ vanishes at the time $t=t_0$ that is given from ~\eqref{dimensionless_x_t} by
\begin{equation}
    t_0 = t_0^{\prime} \frac{m \hbar}{\alpha^2} \approx 0.021 \frac{m \hbar}{\alpha^2} \, .
    \label{t_0_approx}
\end{equation}

We then apply the definition~\eqref{pi_def} to construct the TOA distribution $\pi_0(t)$ at position $x=0$ that corresponds to the Bracken-Melloy state~\eqref{QB_state}. In view of~\eqref{dimensionless_x_t} and~\eqref{current_rel} we can write
\begin{equation}
    \pi_0(t) = \frac{\alpha^2}{m \hbar} \Pi_0(t')
    \label{pi_0_Pi_0_rel}
\end{equation}
in terms of the normalized dimensionless TOA distribution $\Pi_0(t')$ given by
\begin{equation*}
    \Pi_0(t') = \frac{\left\vert
\mathcal{J}\left( 0, t' \right) \right\vert}{\int_{0}^{\infty } \left\vert \mathcal{J} \left( 0,s' \right) \right\vert ds'} \, .
\end{equation*}
Fig.~\ref{TOA_QB_fig} shows the (dimensionless) TOA distribution $\Pi_0(t')$ (solid blue line) as a function of the (dimensionless) time $t'$. As expected, we can readily see that $\Pi_0(t')$ vanishes at the same time $t' = t_0^{\prime}$ as the current, which hence also means in view of~\eqref{pi_0_Pi_0_rel} that $\pi_0(t)$ vanishes at the time $t=t_0$.


\begin{center}
\begin{figure}[!ht]
\centering
\includegraphics[width=0.80\linewidth]{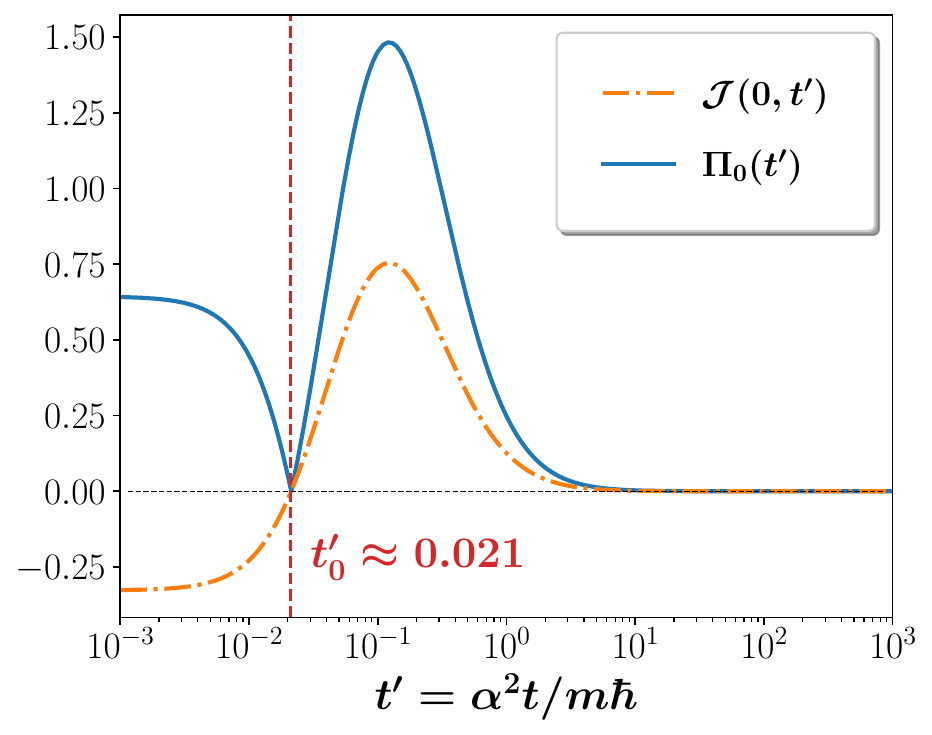}
\caption{\textbf{Current and TOA distribution for the Bracken-Melloy state}(Dimensionless) current $\mathcal{J} \left(0, t' \right)$ (dash-dotted orange line) at position $x=0$ for the Bracken-Melloy state~\eqref{QB_state}, along with the corresponding (dimensionless) TOA distribution $\Pi_0(t')$ (solid blue line) as a function of the (dimensionless) time $t'$, on a log scale.}
\label{TOA_QB_fig}
\end{figure}
\end{center}


The vanishing of $\pi_0(t)$ at time $t=t_0$ is a direct consequence of the change of sign of the current at $t=t_0$. Since this change of sign of the current signals the occurrence of quantum backflow, we can thus also take the vanishing of the TOA distribution $\pi_0(t)$ at time $t=t_0$ as a signature of quantum backflow.

Let us now imagine that we can measure experimentally the times of arrival of the quantum particle at position $x=0$ if it is in the Bracken-Melloy state~\eqref{QB_state}. We assume that our time measurements have a certain uncertainty $\delta t$. We then consider the probability $\mathcal{P}_0(\varepsilon)$ to detect the particle at position $x=0$ within the time interval $t \in [t_0 - \varepsilon, t_0 + \varepsilon]$. By construction of the TOA distribution $\pi_0(t)$, we hence have
\begin{equation*}
    \mathcal{P}_0(\varepsilon) = \int_{t_0 - \varepsilon}^{t_0 + \varepsilon} dt \, \pi_0(t) \, .
\end{equation*}
Since we know that $\pi_0(t)$ vanishes for $t=t_0$, we can thus imagine that the probability $\mathcal{P}_0(\varepsilon)$ can take on very small values even for values of $\varepsilon$ that are much larger than the uncertainty $\delta t$. This would thus provide an observable signature of the occurrence of quantum backflow, which would be accessible \textit{without} having to measure the probability current itself. Indeed, in practice only measurements of the arrival times at the position $x=0$ would be realized: more precisely, one would simply record the time at which a detector, located at the fixed position $x=0$, detects the particle, and repeat such a measurement a large number of times.

We illustrate this idea in the particular case of a $^{87}$Rb rubidium atom, of mass $m_{\text{Rb}} =  144.32 \times 10^{-27} \, \text{kg}$, moving with a speed $v=3 \, \text{mm}/\text{s}$: this allows us to express the momentum scale $\alpha$ as $\alpha = m_{\text{Rb}} v$. This value of $v$ is, in particular, consistent with typical values that can be achieved in experiments using condensates of $^{87}$Rb atoms~\cite{Fabre11, Jendrzejewski12, Dupont23}. In this case, the typical time scale $m_{\text{Rb}} \hbar / \alpha^2$ is given by
\begin{equation}
\frac{m_{\text{Rb}} \hbar}{\alpha^2} = \frac{\hbar}{m_{\text{Rb}} v^2} \approx 0.81 \times 10^{-4} \, \text{s} \, ,
\label{time_scale_Rb}
\end{equation}
and the time $t_0$ at which the TOA distribution vanishes is, in view of~\eqref{t_0_approx} and~\eqref{time_scale_Rb}, given by
\begin{equation*}
    t_0 \approx 1.71 \, \mu\text{s} \, .
\end{equation*}

We now fix a desired value of the probability $\mathcal{P}_0(\varepsilon)$, and compute the corresponding value of $\varepsilon$. We then require the time resolution $\delta t$ of our measurement to be one order of magnitude smaller than $\varepsilon$, i.e. we require $\delta t = \varepsilon / 10$. We repeat this for different values of $\mathcal{P}_0(\varepsilon)$. The results are shown in table~\ref{table2}.

We can immediately see from the values given in table~\ref{table2} that to measure relatively small probabilities $\mathcal{P}_0(\varepsilon)$ in this particular case would require a rather high precision for our time measurements: indeed, the corresponding uncertainties $\delta t$ take values from a tenth of a microsecond [for $\mathcal{P}_0(\varepsilon)=10^{-2}$] to the picosecond [for $\mathcal{P}_0(\varepsilon)=10^{-8}$]. Since the precise relation between the values of $\mathcal{P}_0(\varepsilon)$ and $\varepsilon$ strongly depends on the underlying state, this analysis suggests that the Bracken-Melloy state is not the best suited for the actual experimental implementation of the above scheme for a $^{87}$Rb rubidium atom. Our main objective in this discussion was mainly to illustrate the idea and to provide a proof of principle of how quantum backflow could be experimentally measured by means of TOA measurements, and we keep for further research a dedicated study of a more realistic scenario that would be experimentally feasible.


\begin{table}[h!]
\centering
\vskip 0.2 cm
\begin{tabular}{ |M{1.5cm}|M{3.5cm}|M{3.5cm}|N }
\hline
\vspace{0.1cm}
$\mathcal{P}_0(\varepsilon)$ & $\varepsilon = \frac{m \hbar}{\alpha^2} \varepsilon'$ & $\delta t = \frac{m \hbar}{\alpha^2} \delta t' = \frac{\varepsilon}{10}$ & \\[0.1cm]
\hline\hline
\vspace{0.1cm}
$10^{-2}$ & $1.31 \times 10^{-6} \, \text{s}$ & $1.31 \times 10^{-7} \, \text{s}$ & \\[0.25cm]
$10^{-3}$ & $4.06 \times 10^{-7} \, \text{s}$ & $4.06 \times 10^{-8} \, \text{s}$ & \\[0.25cm]
$10^{-4}$ & $1.24 \times 10^{-7} \, \text{s}$ & $1.24 \times 10^{-8} \, \text{s}$ & \\[0.25cm]
$10^{-5}$ & $3.12 \times 10^{-8} \, \text{s}$ & $3.12 \times 10^{-9} \, \text{s}$ & \\[0.25cm]
$10^{-6}$ & $3.19 \times 10^{-9} \, \text{s}$ & $3.19 \times 10^{-10} \, \text{s}$ & \\[0.25cm]
$10^{-7}$ & $3.19 \times 10^{-10} \, \text{s}$ & $3.19 \times 10^{-11} \, \text{s}$ & \\[0.25cm]
$10^{-8}$ & $3.19 \times 10^{-11} \, \text{s}$ & $3.19 \times 10^{-12} \, \text{s}$ & \\[0.1cm]
\hline
\end{tabular}
\caption{\textbf{Values of the resolution $\delta t$ required for a given value of the probability $\mathcal{P}_0(\varepsilon)$} Values of $\mathcal{P}_0(\varepsilon)$, $\varepsilon$ and $\delta t$ in the particular case of a $^{87}$Rb rubidium atom moving with a typical speed $v=3 \, \text{mm}/\text{s}$.}
\label{table2}
\end{table}


Two final comments can, however, already be made in view of an actual experimental realization of the above scheme to observe quantum backflow by means of TOA measurements:

\begin{enumerate}
    \item First, the state of the system must be relatively easy to prepare, and yet also satisfy the property, crucial for quantum backflow, that it contains only positive momenta at any time during the experiment;

    \item Second, in view of the expression~\eqref{pi_def} of $\pi_x(t)$ in terms of the absolute value $\lvert j_t(x) \rvert$ of the current, the vanishing of $\pi_x(t)$ is in general \textit{not sufficient} in order to ensure that the current changes sign. Indeed, from a mathematical standpoint, nothing a priori prevents the current $j_t(x)$ from being e.g. positive for any $t$ but to nonetheless vanish at some values of $t$ [a current of the form $j_t(x)=1+\cos t$ would be a trivial example that exhibits such a behavior]. Therefore, in an actual experiment, additional conditions about the behavior of the experimentally constructed TOA distribution close to the points where it vanishes are needed in order to unambiguously characterize such points as points where the current changes sign. One such extra condition could for instance be related to the concavity or convexity of $\pi_x(t)$ in the vicinity of such points, or of course, its differentiability if the experimental curve contains enough points. If the experimental state can be theoretically predicted, another possibility could be to apply statistical methods to check quantitatively the agreement between the theoretically predicted TOA distribution and the corresponding experimental distribution.
\end{enumerate}

\subsection{TOA distribution for the free particle with initial superposed Gaussian state}\label{sup_Gaussian_subsec}

\begin{figure*}[!ht]
  \centering
    \includegraphics[scale=0.55]{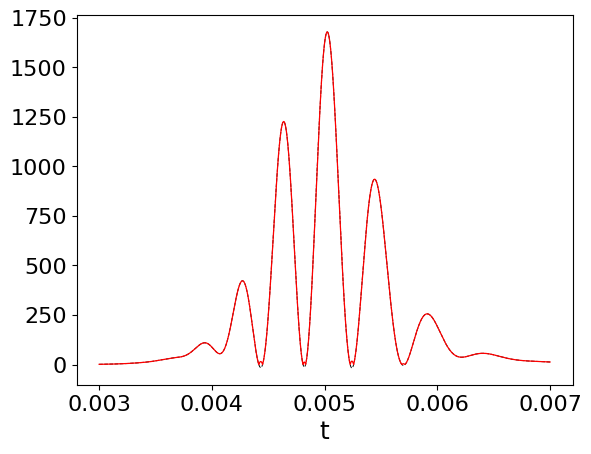}\hspace{0.5cm}\includegraphics[scale=0.55]{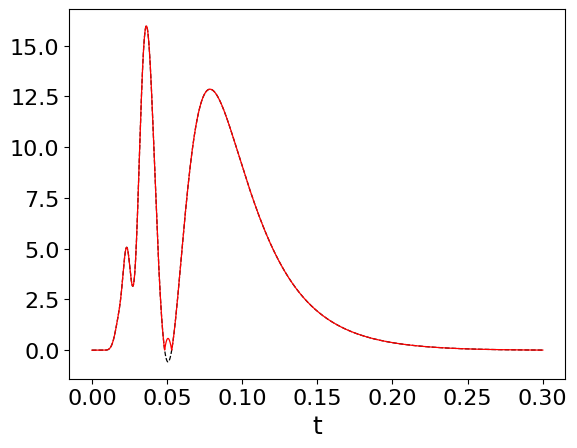}
    \caption{\textbf{Distribution of the TOA for a free particle with an initial
superposition of two Gaussian states} {In this figure, we present the normalized distribution (solid lines) of the time of arrival at position $x = 0$, along with the current density (dashed lines) for a free particle in an initial superposition of two Gaussian states, as defined by equations \eqref{pi_def2} and \eqref{Eq:CurrentGeneral}, along with equations \eqref{Eq:Gauss:Phases}-\eqref{Eq:Gauss:Density}. The parameter values chosen for these illustrations are: $\hbar = 1$, $m = 1$, $\sigma_1 = \sigma_2 = 0.05$, $a_1 = -1$, $a_2 = -0.5$ for both panels. We further take the wave vectors to be $k_1 = 200$ and $k_2 = 100$ for the left panel, versus $k_1 = 5$ and $k_2 = 1.5$ for the right panel.}
} \label{pi(t)}
\end{figure*}

{
In this section, we present an example of derivation of the TOA at a given position distribution for a non-Gaussian system. More precisely, we use our formalism to obtain the TOA distribution for a superposition of two general wave packets and we present a specific application to the case of a non-Gaussian superposition of two Gaussian wave packets.  
}
{
We first consider a wave packet written as a superposition of two wave packets 
\begin{equation}\label{Eq:Superposition:General}
\psi\left( x,t\right) =\sqrt{\mathcal{N}}\left[ \psi_{1}\left( x,t\right) +\psi_{2}\left( x,t\right) \right] ,
\end{equation}%
where $\mathcal{N}$ is the normalization factor. We note that this superposed wave packet is non-Gaussian even if the two underlying wave packets are Gaussian. We then decompose the phase of each wave packet $\psi_{k}=\psi_{k}(x,t)$ into the real and the imaginary phases
\begin{equation}\label{Eq:Wavefunction:General}
    \psi_{j} = e^{\phi_j+i\varphi_j},\ j=1,2
\end{equation}
where $\phi_j=\phi_j(x,t)$ and $\varphi_j=\varphi_j(x,t)$ are two real functions. From the definition of the current (see equation \eqref{Eq:CurrentDef}) we find a general expression for the current of the superposition \eqref{Eq:Superposition:General}
\begin{widetext}
\begin{eqnarray}\label{Eq:CurrentGeneral}
j = \mathcal{N}\left[v_1\rho_1 + v_2\rho_2 + (u_1-u_2)\sqrt{\rho_1\rho_2}\sin(\varphi_1-\varphi_2)+(v_1+v_2)\sqrt{\rho_1\rho_2}\cos(\varphi_1-\varphi_2)\right] ,
\end{eqnarray}
\end{widetext}
where:
\begin{align}
    u_j &= \frac{\hbar}{m}\frac{\partial}{\partial x}\phi_j , \label{Eq:uj} \\
    v_j &= \frac{\hbar}{m}\frac{\partial}{\partial x}\varphi_j \label{Eq:vj}\\
    \rho_j &= |\psi_j|^2 = e^{2\phi_j}. \label{Eq:rhoj}
\end{align}
The first two terms in equation \eqref{Eq:CurrentGeneral} represent the individual current for each wave packet while the next two terms characterize the interference between them. To obtain the TOA distribution, one simply needs to insert the expression of the current \eqref{Eq:CurrentGeneral} into the relation between the TOA and the current \eqref{pi_def2}, up to a constant of normalization.  }
{
We now specialize the analysis to a free particle evolving in one dimension with an initial wave
packet given by the superposition of two Gaussian wave packets with mean
initial positions $a_1$ and $a_2$, initial wave vectors $k_1$ and $k_2$, and initial standard deviations $\sigma_1$ and $\sigma_2$, respectively. In this case, the expression for each wave packet at time $t$ is:
\begin{align}\label{Eq:Wavefct:free}
\psi_{j} &=  \frac{1}{\sqrt{\sqrt{2\pi}\left(\sigma_j+\frac{i\lambda(t)^2}{2\sigma_j}\right)}}e^{-\frac{(x-a_j-2i\sigma_j^2k_j)^2}{2(2\sigma_j^2+i\lambda(t)^2)}-\sigma_j^2k^2+ika_j} \nonumber \\ 
&= C_j\ e^{-\frac{\left( x- x_{j}(t)\right) ^{2}}{4\sigma_{j}(t)^{2}}+i\frac{
\lambda(t)^{2}}{8\sigma_j^{2}\sigma_{j}(t)^{2}}\left( x-x_{j}(t)\right)
^{2} + ik_j x - i\frac{k_j^2\lambda(t)^2}{2}} ,
\end{align}
where $\sigma_j(t)^2 = \sigma_j^2+\frac{\lambda(t)^4}{4\sigma_j^2}$, $\lambda(t) = \sqrt{\frac{\hbar t}{m}}$, the classical position $x_j(t) = a_j + \frac{\hbar}{m} k_j t$, and where the coefficient $C_j = \frac{1}{\sqrt{\sqrt{2\pi}\left(\sigma_j+i\frac{\lambda(t)^2}{2\sigma_j}\right)}} = \frac{1}{(2\pi\sigma_j(t)^2)^{1/4}}e^{-i\chi_j/2}=e^{-\frac{1}{4}\ln{(2\pi\sigma_j(t)^2)}}e^{-\frac{i}{2}\chi_j}$, with $\chi_j = \arccos{\left(\frac{\sigma_j}{\sigma_j(t)}\right)}$. Up to this normalization coefficient $C$, we can rewrite the wave function $\psi_j$ in \eqref{Eq:Wavefct:free} in the general form \eqref{Eq:Wavefunction:General} with:
\begin{align}\label{Eq:Gauss:Phases}
    \phi_j &= -\frac{(x-x_j(t))^2}{4\sigma_j(t)^2}-\frac{1}{4}\ln{(2\pi\sigma_j(t)^2)}\\
    \varphi_j &= \frac{\lambda(t)^2}{8\sigma_j^2\sigma_j(t)^2}(x-x_j(t))^2 + k_j x-\frac{k_j^2\lambda(t)^2}{2} -\frac{\chi_j}{2} 
\end{align}
from which we obtain in view of \eqref{Eq:uj}-\eqref{Eq:vj} the expressions of the velocities
\begin{align}\label{Eq:Gauss:Velocities}
    u_j &= -\frac{\hbar}{m}\frac{x-x_j(t)}{2\sigma_j(t)^2}\\
    v_j &= \frac{\hbar}{m}\frac{\lambda(t)^2}{4\sigma_j^2\sigma_j(t)^2}(x-x_j(t)) + \frac{\hbar}{m}k_j
\end{align}
and in view of \eqref{Eq:rhoj} the expression of the densities
\begin{equation}\label{Eq:Gauss:Density}
    \rho_j =\frac{1}{\sqrt{2\pi\sigma_j(t)^2}}e^{-\frac{(x-x_j(t))^2}{2\sigma_j(t)^2}} .
\end{equation}
}

{
For illustration purposes, we show in Fig. \ref{pi(t)} the distribution of the TOA for this system with a chosen set of parameter values. Note that the distribution depicted in Fig. \ref{pi(t)} has been normalized and therefore represents the probability distribution of the time-of-arrival for particles reaching the detector at a time $t \geq 0$. The two panels in Fig. \ref{pi(t)} use the exact same parameters as Figures 2 and 3 in \cite{Galapon05}, which analyze Kijowski's density distribution (see also Figure 3 in \cite{Hegerfeldt03} and Figure 3 in \cite{Galapon16} for further analysis of the time-distribution of a linear superposition of two free Gaussian wave packets). From a comparison with Figures 2 and 3 in \cite{Galapon05}, we observe that the distributions we obtain are qualitatively similar to Kijowski's distributions despite small quantitative differences, such as the maximum values of the peaks in both figures. Notably, in the right panel of Fig. \ref{pi(t)}, the left peak is higher than the right peak, in contrast to Figure 3 in \cite{Galapon05}. We also observe the presence of a kink at $t \approx 0.03$, which does not exist in the aforementioned figure. Additionally, we note that the current density becomes negative around $t = 0.05$. In coherence with a result noted in the caption of Figure 2 in \cite{Galapon05}, the left panel in Fig. \ref{pi(t)} shows negative values of the current density in small regions around $t \approx 0.0044,\ 0.0048,\ 0.0053$, which indicates the presence of a backflow effect. This occurs when the wave vectors $k_1 = 200$ and $k_2 = 100$ are large compared to the standard deviation $1/(2\sigma_1) = 1/(2\sigma_2) = 10$ of the wave packets.}
{
Interestingly, our result \eqref{Eq:CurrentGeneral}, along with equations \eqref{Eq:Gauss:Phases}-\eqref{Eq:Gauss:Density} for the specific case with $k_1 = -k_2$, is consistent with the findings of Leavens \cite{leavens1998}, who also analyzes the TOA distribution for a free particle with an initial state given by a superposition of two Gaussian wave packets with opposite momenta. This can be attributed to the fact that we employ the same relation between the TOA distribution and the probability current (see equation \eqref{pi_def2} above, and equation (10) in \cite{leavens1998}, up to a normalization factor). As indicated before, the noticeable difference is that we obtain the same expression without having to refer to the Bohmian interpretation of quantum mechanics. }
{
Note that in this section we are able to obtain a general closed-form expression for the current density of a superposition for any wave packets, Gaussian or otherwise (see equation \eqref{Eq:CurrentGeneral}). This greatly simplifies numerical calculations since it eliminates the need to compute the derivative of the wave function with respect to position $x$. A future application of this formula could involve calculating the time distribution for a superposition of Gaussian wave packets in the presence of gravity and/or in a harmonic trap.
}

\section{Conclusions and suggestions for further research}\label{concl_sec}

This paper presents a general framework for the analysis of time measurements for continuous quantum systems and discusses examples of applications. Our
approach, which is based on a straightforward stochastic representation of quantum measurements, is general enough to be applicable in principle to any observable in any continuous system. The extension to TOA distributions for quantum systems with discrete state spaces, which requires a related but different methodology, will be presented in future work.

We also discuss, in particular, the proof of principle of a promising route offered by our TOA distribution~\eqref{pi_def} towards the experimental observation of the still elusive phenomenon of quantum backflow. Since the latter is also known to occur in the presence of a linear potential~\cite{Melloy98_1}, free-fall experiments can for instance be natural candidates.

Given that quantum tunneling can be regarded as a time-of-arrival problem, our approach could for instance be used to analyze the escape time from a region for a particle facing a potential barrier with a peak corresponding to an energy higher than that carried by the particle. The issue of tunneling time has received substantial attention in the literature, but no consensus has been reached in the absence of a straightforward method for handling the question with the basic axioms from quantum mechanics. Another potentially fruitful application would involve analyzing the energy transition time for a system with a time-dependent Hamiltonian, such as the time-dependent harmonic oscillator.

\textbf{Acknowledgments.} Part of his research was conducted while LM was a visiting professor at MIT. He would like to thank Prof. Seth Lloyd for helpful comments, and Prof. Lorenzo Maccone for his invitation to give a seminar at the University of Pavia and for providing fruitful discussions on this occasion. We also acknowledge useful feedback from participants at the PAFT 2024 conference.

\bibliography{Beau_PRA_v2.bib}

\providecommand{\noopsort}[1]{}\providecommand{\singleletter}[1]{#1}%
\begin{thebibliography}{50}%
\makeatletter
\providecommand \@ifxundefined [1]{%
 \@ifx{#1\undefined}
}%
\providecommand \@ifnum [1]{%
 \ifnum #1\expandafter \@firstoftwo
 \else \expandafter \@secondoftwo
 \fi
}%
\providecommand \@ifx [1]{%
 \ifx #1\expandafter \@firstoftwo
 \else \expandafter \@secondoftwo
 \fi
}%
\providecommand \natexlab [1]{#1}%
\providecommand \enquote  [1]{``#1''}%
\providecommand \bibnamefont  [1]{#1}%
\providecommand \bibfnamefont [1]{#1}%
\providecommand \citenamefont [1]{#1}%
\providecommand \href@noop [0]{\@secondoftwo}%
\providecommand \href [0]{\begingroup \@sanitize@url \@href}%
\providecommand \@href[1]{\@@startlink{#1}\@@href}%
\providecommand \@@href[1]{\endgroup#1\@@endlink}%
\providecommand \@sanitize@url [0]{\catcode `\\12\catcode `\$12\catcode
  `\&12\catcode `\#12\catcode `\^12\catcode `\_12\catcode `\%12\relax}%
\providecommand \@@startlink[1]{}%
\providecommand \@@endlink[0]{}%
\providecommand \url  [0]{\begingroup\@sanitize@url \@url }%
\providecommand \@url [1]{\endgroup\@href {#1}{\urlprefix }}%
\providecommand \urlprefix  [0]{URL }%
\providecommand \Eprint [0]{\href }%
\providecommand \doibase [0]{https://doi.org/}%
\providecommand \selectlanguage [0]{\@gobble}%
\providecommand \bibinfo  [0]{\@secondoftwo}%
\providecommand \bibfield  [0]{\@secondoftwo}%
\providecommand \translation [1]{[#1]}%
\providecommand \BibitemOpen [0]{}%
\providecommand \bibitemStop [0]{}%
\providecommand \bibitemNoStop [0]{.\EOS\space}%
\providecommand \EOS [0]{\spacefactor3000\relax}%
\providecommand \BibitemShut  [1]{\csname bibitem#1\endcsname}%
\let\auto@bib@innerbib\@empty
\bibitem [{\citenamefont {Pauli}\ \emph {et~al.}(1933)\citenamefont {Pauli}
  \emph {et~al.}}]{pauli1933handbuch}%
  \BibitemOpen
  \bibfield  {author} {\bibinfo {author} {\bibfnamefont {W.}~\bibnamefont
  {Pauli}} \emph {et~al.},\ }\bibfield  {title} {\bibinfo {title} {Handbuch der
  physik},\ }\href@noop {} {\bibfield  {journal} {\bibinfo  {journal} {Geiger
  and scheel}\ }\textbf {\bibinfo {volume} {2}},\ \bibinfo {pages} {83}
  (\bibinfo {year} {1933})}\BibitemShut {NoStop}%
\bibitem [{\citenamefont {Copley}\ and\ \citenamefont
  {Udovic}(1993)}]{udovic1993neutron}%
  \BibitemOpen
  \bibfield  {author} {\bibinfo {author} {\bibfnamefont {J.~R.~D.}\
  \bibnamefont {Copley}}\ and\ \bibinfo {author} {\bibfnamefont {T.~J.}\
  \bibnamefont {Udovic}},\ }\bibfield  {title} {\bibinfo {title} {Neutron
  time-of-flight spectroscopy},\ }\href
  {https://www.ncbi.nlm.nih.gov/pmc/articles/PMC4927250/#:~:text=Time%2Dof%2Dflight%20spectroscopy%20is,the%20associated%20quasielastic%20neutron%20scattering.}
  {\bibfield  {journal} {\bibinfo  {journal} {J Res Natl Inst Stand Technol.}\
  }\textbf {\bibinfo {volume} {98}} (\bibinfo {year} {1993})}\BibitemShut
  {NoStop}%
\bibitem [{\citenamefont {Kurtsiefer}\ \emph {et~al.}(1995)\citenamefont
  {Kurtsiefer}, \citenamefont {Pfau}, \citenamefont {Ekstrom},\ and\
  \citenamefont {Mlynek}}]{kurtsiefer1995time}%
  \BibitemOpen
  \bibfield  {author} {\bibinfo {author} {\bibfnamefont {C.}~\bibnamefont
  {Kurtsiefer}}, \bibinfo {author} {\bibfnamefont {T.}~\bibnamefont {Pfau}},
  \bibinfo {author} {\bibfnamefont {C.~R.}\ \bibnamefont {Ekstrom}},\ and\
  \bibinfo {author} {\bibfnamefont {J.}~\bibnamefont {Mlynek}},\ }\bibfield
  {title} {\bibinfo {title} {Time-resolved detection of atoms diffracted from a
  standing light wave},\ }\href {https://doi.org/10.1007/BF01135866} {\bibfield
   {journal} {\bibinfo  {journal} {Applied Physics B}\ }\textbf {\bibinfo
  {volume} {60}},\ \bibinfo {pages} {229} (\bibinfo {year} {1995})}\BibitemShut
  {NoStop}%
\bibitem [{\citenamefont {Kothe}\ \emph {et~al.}(2013)\citenamefont {Kothe},
  \citenamefont {Metje}, \citenamefont {Wilke}, \citenamefont {Moguilevski},
  \citenamefont {Engel}, \citenamefont {Al-Obaidi}, \citenamefont {Richter},
  \citenamefont {Golnak}, \citenamefont {Kiyan},\ and\ \citenamefont
  {Aziz}}]{kothe2013time}%
  \BibitemOpen
  \bibfield  {author} {\bibinfo {author} {\bibfnamefont {A.}~\bibnamefont
  {Kothe}}, \bibinfo {author} {\bibfnamefont {J.}~\bibnamefont {Metje}},
  \bibinfo {author} {\bibfnamefont {M.}~\bibnamefont {Wilke}}, \bibinfo
  {author} {\bibfnamefont {A.}~\bibnamefont {Moguilevski}}, \bibinfo {author}
  {\bibfnamefont {N.}~\bibnamefont {Engel}}, \bibinfo {author} {\bibfnamefont
  {R.}~\bibnamefont {Al-Obaidi}}, \bibinfo {author} {\bibfnamefont
  {C.}~\bibnamefont {Richter}}, \bibinfo {author} {\bibfnamefont
  {R.}~\bibnamefont {Golnak}}, \bibinfo {author} {\bibfnamefont {I.~Y.}\
  \bibnamefont {Kiyan}},\ and\ \bibinfo {author} {\bibfnamefont {E.~F.}\
  \bibnamefont {Aziz}},\ }\bibfield  {title} {\bibinfo {title} {Time-of-flight
  electron spectrometer for a broad range of kinetic energies},\ }\href@noop {}
  {\bibfield  {journal} {\bibinfo  {journal} {Review of Scientific
  Instruments}\ }\textbf {\bibinfo {volume} {84}} (\bibinfo {year}
  {2013})}\BibitemShut {NoStop}%
\bibitem [{\citenamefont {Dufour}\ \emph {et~al.}(2014)\citenamefont {Dufour},
  \citenamefont {Debu}, \citenamefont {Lambrecht}, \citenamefont
  {Nesvizhevsky}, \citenamefont {Reynaud},\ and\ \citenamefont
  {Voronin}}]{dufour2014shaping}%
  \BibitemOpen
  \bibfield  {author} {\bibinfo {author} {\bibfnamefont {G.}~\bibnamefont
  {Dufour}}, \bibinfo {author} {\bibfnamefont {P.}~\bibnamefont {Debu}},
  \bibinfo {author} {\bibfnamefont {A.}~\bibnamefont {Lambrecht}}, \bibinfo
  {author} {\bibfnamefont {V.~V.}\ \bibnamefont {Nesvizhevsky}}, \bibinfo
  {author} {\bibfnamefont {S.}~\bibnamefont {Reynaud}},\ and\ \bibinfo {author}
  {\bibfnamefont {A.~Y.}\ \bibnamefont {Voronin}},\ }\bibfield  {title}
  {\bibinfo {title} {Shaping the distribution of vertical velocities of
  antihydrogen in gbar},\ }\href
  {https://doi.org/10.1140/epjc/s10052-014-2731-8} {\bibfield  {journal}
  {\bibinfo  {journal} {The European Physical Journal C}\ }\textbf {\bibinfo
  {volume} {74}},\ \bibinfo {pages} {2731} (\bibinfo {year}
  {2014})}\BibitemShut {NoStop}%
\bibitem [{\citenamefont {Gliserin}\ \emph {et~al.}(2016)\citenamefont
  {Gliserin}, \citenamefont {Walbran},\ and\ \citenamefont
  {Baum}}]{gliserin2016high}%
  \BibitemOpen
  \bibfield  {author} {\bibinfo {author} {\bibfnamefont {A.}~\bibnamefont
  {Gliserin}}, \bibinfo {author} {\bibfnamefont {M.}~\bibnamefont {Walbran}},\
  and\ \bibinfo {author} {\bibfnamefont {P.}~\bibnamefont {Baum}},\ }\bibfield
  {title} {\bibinfo {title} {A high-resolution time-of-flight energy analyzer
  for femtosecond electron pulses at 30 kev},\ }\href@noop {} {\bibfield
  {journal} {\bibinfo  {journal} {Review of Scientific Instruments}\ }\textbf
  {\bibinfo {volume} {87}} (\bibinfo {year} {2016})}\BibitemShut {NoStop}%
\bibitem [{\citenamefont {Rousselle}\ \emph {et~al.}(2022)\citenamefont
  {Rousselle}, \citenamefont {Cladé}, \citenamefont {Guellati-Khelifa},
  \citenamefont {Guérout},\ and\ \citenamefont
  {Reynaud}}]{rousselle2022analysis}%
  \BibitemOpen
  \bibfield  {author} {\bibinfo {author} {\bibfnamefont {O.}~\bibnamefont
  {Rousselle}}, \bibinfo {author} {\bibfnamefont {P.}~\bibnamefont {Cladé}},
  \bibinfo {author} {\bibfnamefont {S.}~\bibnamefont {Guellati-Khelifa}},
  \bibinfo {author} {\bibfnamefont {R.}~\bibnamefont {Guérout}},\ and\
  \bibinfo {author} {\bibfnamefont {S.}~\bibnamefont {Reynaud}},\ }\bibfield
  {title} {\bibinfo {title} {Analysis of the timing of freely falling
  antihydrogen},\ }\href {https://doi.org/10.1088/1367-2630/ac5b57} {\bibfield
  {journal} {\bibinfo  {journal} {New Journal of Physics}\ }\textbf {\bibinfo
  {volume} {24}},\ \bibinfo {pages} {033045} (\bibinfo {year}
  {2022})}\BibitemShut {NoStop}%
\bibitem [{\citenamefont {Aharonov}\ and\ \citenamefont
  {Bohm}(1961)}]{aharonov1961time}%
  \BibitemOpen
  \bibfield  {author} {\bibinfo {author} {\bibfnamefont {Y.}~\bibnamefont
  {Aharonov}}\ and\ \bibinfo {author} {\bibfnamefont {D.}~\bibnamefont
  {Bohm}},\ }\bibfield  {title} {\bibinfo {title} {Time in the quantum theory
  and the uncertainty relation for time and energy},\ }\href@noop {} {\bibfield
   {journal} {\bibinfo  {journal} {Physical Review}\ }\textbf {\bibinfo
  {volume} {122}},\ \bibinfo {pages} {1649} (\bibinfo {year}
  {1961})}\BibitemShut {NoStop}%
\bibitem [{\citenamefont {Giannitrapani}(1997)}]{giannitrapani1997positive}%
  \BibitemOpen
  \bibfield  {author} {\bibinfo {author} {\bibfnamefont {R.}~\bibnamefont
  {Giannitrapani}},\ }\bibfield  {title} {\bibinfo {title}
  {Positive-operator-valued time observable in quantum mechanics},\ }\href
  {https://doi.org/10.1007/BF02435757} {\bibfield  {journal} {\bibinfo
  {journal} {International Journal of Theoretical Physics}\ }\textbf {\bibinfo
  {volume} {36}},\ \bibinfo {pages} {1575} (\bibinfo {year}
  {1997})}\BibitemShut {NoStop}%
\bibitem [{\citenamefont {Delgado}(1999)}]{delgado1999quantum}%
  \BibitemOpen
  \bibfield  {author} {\bibinfo {author} {\bibfnamefont {V.}~\bibnamefont
  {Delgado}},\ }\bibfield  {title} {\bibinfo {title} {Quantum probability
  distribution of arrival times and probability current density},\ }\href
  {https://doi.org/10.1103/PhysRevA.59.1010} {\bibfield  {journal} {\bibinfo
  {journal} {Phys. Rev. A}\ }\textbf {\bibinfo {volume} {59}},\ \bibinfo
  {pages} {1010} (\bibinfo {year} {1999})}\BibitemShut {NoStop}%
\bibitem [{\citenamefont {Galapon}(2004)}]{galapon2004shouldn}%
  \BibitemOpen
  \bibfield  {author} {\bibinfo {author} {\bibfnamefont {E.~A.}\ \bibnamefont
  {Galapon}},\ }\bibfield  {title} {\bibinfo {title} {Shouldn’t there be an
  antithesis to quantization?},\ }\href@noop {} {\bibfield  {journal} {\bibinfo
   {journal} {Journal of mathematical physics}\ }\textbf {\bibinfo {volume}
  {45}},\ \bibinfo {pages} {3180} (\bibinfo {year} {2004})}\BibitemShut
  {NoStop}%
\bibitem [{\citenamefont {Vona}\ \emph {et~al.}(2013)\citenamefont {Vona},
  \citenamefont {Hinrichs},\ and\ \citenamefont {D\"urr}}]{vona2013does}%
  \BibitemOpen
  \bibfield  {author} {\bibinfo {author} {\bibfnamefont {N.}~\bibnamefont
  {Vona}}, \bibinfo {author} {\bibfnamefont {G.}~\bibnamefont {Hinrichs}},\
  and\ \bibinfo {author} {\bibfnamefont {D.}~\bibnamefont {D\"urr}},\
  }\bibfield  {title} {\bibinfo {title} {What does one measure when one
  measures the arrival time of a quantum particle?},\ }\href
  {https://doi.org/10.1103/PhysRevLett.111.220404} {\bibfield  {journal}
  {\bibinfo  {journal} {Phys. Rev. Lett.}\ }\textbf {\bibinfo {volume} {111}},\
  \bibinfo {pages} {220404} (\bibinfo {year} {2013})}\BibitemShut {NoStop}%
\bibitem [{\citenamefont {Tumulka}(2022)}]{TUMULKA2022168910}%
  \BibitemOpen
  \bibfield  {author} {\bibinfo {author} {\bibfnamefont {R.}~\bibnamefont
  {Tumulka}},\ }\bibfield  {title} {\bibinfo {title} {Distribution of the time
  at which an ideal detector clicks},\ }\href
  {https://doi.org/https://doi.org/10.1016/j.aop.2022.168910} {\bibfield
  {journal} {\bibinfo  {journal} {Annals of Physics}\ }\textbf {\bibinfo
  {volume} {442}},\ \bibinfo {pages} {168910} (\bibinfo {year}
  {2022})}\BibitemShut {NoStop}%
\bibitem [{\citenamefont {Grot}\ \emph {et~al.}(1996)\citenamefont {Grot},
  \citenamefont {Rovelli},\ and\ \citenamefont {Tate}}]{Rovelli96}%
  \BibitemOpen
  \bibfield  {author} {\bibinfo {author} {\bibfnamefont {N.}~\bibnamefont
  {Grot}}, \bibinfo {author} {\bibfnamefont {C.}~\bibnamefont {Rovelli}},\ and\
  \bibinfo {author} {\bibfnamefont {R.~S.}\ \bibnamefont {Tate}},\ }\bibfield
  {title} {\bibinfo {title} {Time of arrival in quantum mechanics},\ }\href
  {https://doi.org/10.1103/PhysRevA.54.4676} {\bibfield  {journal} {\bibinfo
  {journal} {Phys. Rev. A}\ }\textbf {\bibinfo {volume} {54}},\ \bibinfo
  {pages} {4676} (\bibinfo {year} {1996})}\BibitemShut {NoStop}%
\bibitem [{\citenamefont {Flores}\ and\ \citenamefont
  {Galapon}(2019)}]{flores2019quantum}%
  \BibitemOpen
  \bibfield  {author} {\bibinfo {author} {\bibfnamefont {P.~C.~M.}\
  \bibnamefont {Flores}}\ and\ \bibinfo {author} {\bibfnamefont {E.~A.}\
  \bibnamefont {Galapon}},\ }\bibfield  {title} {\bibinfo {title} {Quantum
  free-fall motion and quantum violation of the weak equivalence principle},\
  }\href {https://doi.org/10.1103/PhysRevA.99.042113} {\bibfield  {journal}
  {\bibinfo  {journal} {Phys. Rev. A}\ }\textbf {\bibinfo {volume} {99}},\
  \bibinfo {pages} {042113} (\bibinfo {year} {2019})}\BibitemShut {NoStop}%
\bibitem [{\citenamefont {Leavens}(1993)}]{leavens1993arrival}%
  \BibitemOpen
  \bibfield  {author} {\bibinfo {author} {\bibfnamefont {C.}~\bibnamefont
  {Leavens}},\ }\bibfield  {title} {\bibinfo {title} {Arrival time
  distributions},\ }\href
  {https://doi.org/https://doi.org/10.1016/0375-9601(93)90722-C} {\bibfield
  {journal} {\bibinfo  {journal} {Physics Letters A}\ }\textbf {\bibinfo
  {volume} {178}},\ \bibinfo {pages} {27} (\bibinfo {year} {1993})}\BibitemShut
  {NoStop}%
\bibitem [{\citenamefont {McKinnon}\ and\ \citenamefont
  {Leavens}(1995)}]{mckinnon1995distributions}%
  \BibitemOpen
  \bibfield  {author} {\bibinfo {author} {\bibfnamefont {W.~R.}\ \bibnamefont
  {McKinnon}}\ and\ \bibinfo {author} {\bibfnamefont {C.~R.}\ \bibnamefont
  {Leavens}},\ }\bibfield  {title} {\bibinfo {title} {Distributions of delay
  times and transmission times in bohm's causal interpretation of quantum
  mechanics},\ }\href {https://doi.org/10.1103/PhysRevA.51.2748} {\bibfield
  {journal} {\bibinfo  {journal} {Phys. Rev. A}\ }\textbf {\bibinfo {volume}
  {51}},\ \bibinfo {pages} {2748} (\bibinfo {year} {1995})}\BibitemShut
  {NoStop}%
\bibitem [{\citenamefont {Daumer}\ \emph {et~al.}(1997)\citenamefont {Daumer},
  \citenamefont {D{\"u}rr}, \citenamefont {Goldstein},\ and\ \citenamefont
  {Zanghi}}]{daumer1997quantum}%
  \BibitemOpen
  \bibfield  {author} {\bibinfo {author} {\bibfnamefont {M.}~\bibnamefont
  {Daumer}}, \bibinfo {author} {\bibfnamefont {D.}~\bibnamefont {D{\"u}rr}},
  \bibinfo {author} {\bibfnamefont {S.}~\bibnamefont {Goldstein}},\ and\
  \bibinfo {author} {\bibfnamefont {N.}~\bibnamefont {Zanghi}},\ }\bibfield
  {title} {\bibinfo {title} {On the quantum probability flux through
  surfaces},\ }\href {https://doi.org/10.1023/B:JOSS.0000015181.86864.fb}
  {\bibfield  {journal} {\bibinfo  {journal} {Journal of Statistical Physics}\
  }\textbf {\bibinfo {volume} {88}},\ \bibinfo {pages} {967} (\bibinfo {year}
  {1997})}\BibitemShut {NoStop}%
\bibitem [{\citenamefont {Leavens}(1998)}]{leavens1998}%
  \BibitemOpen
  \bibfield  {author} {\bibinfo {author} {\bibfnamefont {C.~R.}\ \bibnamefont
  {Leavens}},\ }\bibfield  {title} {\bibinfo {title} {Time of arrival in
  quantum and bohmian mechanics},\ }\href
  {https://doi.org/10.1103/PhysRevA.58.840} {\bibfield  {journal} {\bibinfo
  {journal} {Phys. Rev. A}\ }\textbf {\bibinfo {volume} {58}},\ \bibinfo
  {pages} {840} (\bibinfo {year} {1998})}\BibitemShut {NoStop}%
\bibitem [{\citenamefont {Beau}\ and\ \citenamefont
  {Martellini}(2024{\natexlab{a}})}]{Beau24}%
  \BibitemOpen
  \bibfield  {author} {\bibinfo {author} {\bibfnamefont {M.}~\bibnamefont
  {Beau}}\ and\ \bibinfo {author} {\bibfnamefont {L.}~\bibnamefont
  {Martellini}},\ }\bibfield  {title} {\bibinfo {title} {Quantum delay in the
  time of arrival of free-falling atoms},\ }\href
  {https://doi.org/10.1103/PhysRevA.109.012216} {\bibfield  {journal} {\bibinfo
   {journal} {Phys. Rev. A}\ }\textbf {\bibinfo {volume} {109}},\ \bibinfo
  {pages} {012216} (\bibinfo {year} {2024}{\natexlab{a}})}\BibitemShut
  {NoStop}%
\bibitem [{\citenamefont {Beau}\ and\ \citenamefont
  {Martellini}(2024{\natexlab{b}})}]{beau2024mutual}%
  \BibitemOpen
  \bibfield  {author} {\bibinfo {author} {\bibfnamefont {M.}~\bibnamefont
  {Beau}}\ and\ \bibinfo {author} {\bibfnamefont {L.}~\bibnamefont
  {Martellini}},\ }\bibfield  {title} {\bibinfo {title} {On the mutual
  exclusiveness of time and position in quantum physics and the corresponding
  uncertainty relation for free falling particles},\ }\href
  {https://arxiv.org/abs/2403.06057} {\bibfield  {journal} {\bibinfo  {journal}
  {arXiv preprint arXiv:2403.06057}\ } (\bibinfo {year}
  {2024}{\natexlab{b}})}\BibitemShut {NoStop}%
\bibitem [{\citenamefont {Bohm}(1952)}]{bohm1952suggested}%
  \BibitemOpen
  \bibfield  {author} {\bibinfo {author} {\bibfnamefont {D.}~\bibnamefont
  {Bohm}},\ }\bibfield  {title} {\bibinfo {title} {A suggested interpretation
  of the quantum theory in terms of "hidden" variables. i},\ }\href
  {https://doi.org/10.1103/PhysRev.85.166} {\bibfield  {journal} {\bibinfo
  {journal} {Phys. Rev.}\ }\textbf {\bibinfo {volume} {85}},\ \bibinfo {pages}
  {166} (\bibinfo {year} {1952})}\BibitemShut {NoStop}%
\bibitem [{\citenamefont {Allcock}(1969)}]{Allcock69III}%
  \BibitemOpen
  \bibfield  {author} {\bibinfo {author} {\bibfnamefont {G.}~\bibnamefont
  {Allcock}},\ }\bibfield  {title} {\bibinfo {title} {The time of arrival in
  quantum mechanics {III}. {T}he measurement ensemble},\ }\href
  {https://doi.org/https://doi.org/10.1016/0003-4916(69)90253-X} {\bibfield
  {journal} {\bibinfo  {journal} {Ann. Phys. (N. Y.)}\ }\textbf {\bibinfo
  {volume} {53}},\ \bibinfo {pages} {311} (\bibinfo {year} {1969})}\BibitemShut
  {NoStop}%
\bibitem [{\citenamefont {Shorack}\ and\ \citenamefont
  {Shorack}(2000)}]{shorack2000probability}%
  \BibitemOpen
  \bibfield  {author} {\bibinfo {author} {\bibfnamefont {G.~R.}\ \bibnamefont
  {Shorack}}\ and\ \bibinfo {author} {\bibfnamefont {G.}~\bibnamefont
  {Shorack}},\ }\href@noop {} {\emph {\bibinfo {title} {Probability for
  statisticians}}},\ Vol.\ \bibinfo {volume} {951}\ (\bibinfo  {publisher}
  {Springer},\ \bibinfo {year} {2000})\BibitemShut {NoStop}%
\bibitem [{\citenamefont {Pishro-Niks}(2014)}]{StochTextbook}%
  \BibitemOpen
  \bibfield  {author} {\bibinfo {author} {\bibfnamefont {H.}~\bibnamefont
  {Pishro-Niks}},\ }\href@noop {} {\emph {\bibinfo {title} {Introduction to
  probability, statistics, and random processes}}}\ (\bibinfo  {publisher}
  {Kappa Research LLC},\ \bibinfo {year} {2014})\BibitemShut {NoStop}%
\bibitem [{\citenamefont {Landau}\ and\ \citenamefont
  {Lifshitz}(1977)}]{LandauQM77}%
  \BibitemOpen
  \bibfield  {author} {\bibinfo {author} {\bibfnamefont {L.}~\bibnamefont
  {Landau}}\ and\ \bibinfo {author} {\bibfnamefont {E.}~\bibnamefont
  {Lifshitz}},\ }\bibinfo {title} {Course of theoretical physics vol. 3:
  Quantum mechanics}\ (\bibinfo  {publisher} {Pergamon Press},\ \bibinfo {year}
  {1977})\ Chap.~\bibinfo {chapter} {XV}, p.\ \bibinfo {pages} {115},\ \bibinfo
  {edition} {3rd}\ ed.\BibitemShut {Stop}%
\bibitem [{\citenamefont {Bracken}\ and\ \citenamefont
  {Melloy}(1994)}]{Bracken94}%
  \BibitemOpen
  \bibfield  {author} {\bibinfo {author} {\bibfnamefont {A.~J.}\ \bibnamefont
  {Bracken}}\ and\ \bibinfo {author} {\bibfnamefont {G.~F.}\ \bibnamefont
  {Melloy}},\ }\bibfield  {title} {\bibinfo {title} {Probability backflow and a
  new dimensionless quantum number},\ }\href
  {https://doi.org/10.1088/0305-4470/27/6/040} {\bibfield  {journal} {\bibinfo
  {journal} {J. Phys. A: Math. Gen.}\ }\textbf {\bibinfo {volume} {27}},\
  \bibinfo {pages} {2197} (\bibinfo {year} {1994})}\BibitemShut {NoStop}%
\bibitem [{\citenamefont {Melloy}\ and\ \citenamefont
  {Bracken}(1998{\natexlab{a}})}]{Melloy98_1}%
  \BibitemOpen
  \bibfield  {author} {\bibinfo {author} {\bibfnamefont {G.~F.}\ \bibnamefont
  {Melloy}}\ and\ \bibinfo {author} {\bibfnamefont {A.~J.}\ \bibnamefont
  {Bracken}},\ }\bibfield  {title} {\bibinfo {title} {The velocity of
  probability transport in quantum mechanics},\ }\href
  {https://doi.org/https://doi.org/10.1002/andp.199851007-818} {\bibfield
  {journal} {\bibinfo  {journal} {Ann. Phys. (Leipzig)}\ }\textbf {\bibinfo
  {volume} {510}},\ \bibinfo {pages} {726} (\bibinfo {year}
  {1998}{\natexlab{a}})}\BibitemShut {NoStop}%
\bibitem [{\citenamefont {Melloy}\ and\ \citenamefont
  {Bracken}(1998{\natexlab{b}})}]{Melloy98_2}%
  \BibitemOpen
  \bibfield  {author} {\bibinfo {author} {\bibfnamefont {G.~F.}\ \bibnamefont
  {Melloy}}\ and\ \bibinfo {author} {\bibfnamefont {A.~J.}\ \bibnamefont
  {Bracken}},\ }\bibfield  {title} {\bibinfo {title} {Probability backflow for
  a {D}irac particle},\ }\href
  {https://doi.org/https://doi.org/10.1023/A:1018724313788} {\bibfield
  {journal} {\bibinfo  {journal} {Found. Phys.}\ }\textbf {\bibinfo {volume}
  {28}},\ \bibinfo {pages} {505} (\bibinfo {year}
  {1998}{\natexlab{b}})}\BibitemShut {NoStop}%
\bibitem [{\citenamefont {Ashfaque}\ \emph {et~al.}(2019)\citenamefont
  {Ashfaque}, \citenamefont {Lynch},\ and\ \citenamefont
  {Strange}}]{Ashfaque19}%
  \BibitemOpen
  \bibfield  {author} {\bibinfo {author} {\bibfnamefont {J.}~\bibnamefont
  {Ashfaque}}, \bibinfo {author} {\bibfnamefont {J.}~\bibnamefont {Lynch}},\
  and\ \bibinfo {author} {\bibfnamefont {P.}~\bibnamefont {Strange}},\
  }\bibfield  {title} {\bibinfo {title} {Relativistic quantum backflow},\
  }\href {https://doi.org/10.1088/1402-4896/ab265c} {\bibfield  {journal}
  {\bibinfo  {journal} {Phys. Scr.}\ }\textbf {\bibinfo {volume} {94}},\
  \bibinfo {pages} {125107} (\bibinfo {year} {2019})}\BibitemShut {NoStop}%
\bibitem [{\citenamefont {Barbier}(2020)}]{Barbier20}%
  \BibitemOpen
  \bibfield  {author} {\bibinfo {author} {\bibfnamefont {M.}~\bibnamefont
  {Barbier}},\ }\bibfield  {title} {\bibinfo {title} {Quantum backflow for
  many-particle systems},\ }\href {https://doi.org/10.1103/PhysRevA.102.023334}
  {\bibfield  {journal} {\bibinfo  {journal} {Phys. Rev. A}\ }\textbf {\bibinfo
  {volume} {102}},\ \bibinfo {pages} {023334} (\bibinfo {year}
  {2020})}\BibitemShut {NoStop}%
\bibitem [{\citenamefont {Strange}(2012)}]{Strange12}%
  \BibitemOpen
  \bibfield  {author} {\bibinfo {author} {\bibfnamefont {P.}~\bibnamefont
  {Strange}},\ }\bibfield  {title} {\bibinfo {title} {Large quantum probability
  backflow and the azimuthal angle–angular momentum uncertainty relation for
  an electron in a constant magnetic field},\ }\href
  {https://doi.org/10.1088/0143-0807/33/5/1147} {\bibfield  {journal} {\bibinfo
   {journal} {Eur. J. Phys.}\ }\textbf {\bibinfo {volume} {33}},\ \bibinfo
  {pages} {1147} (\bibinfo {year} {2012})}\BibitemShut {NoStop}%
\bibitem [{\citenamefont {Barbier}\ \emph {et~al.}(2023)\citenamefont
  {Barbier}, \citenamefont {Goussev},\ and\ \citenamefont
  {Srivastava}}]{Barbier23}%
  \BibitemOpen
  \bibfield  {author} {\bibinfo {author} {\bibfnamefont {M.}~\bibnamefont
  {Barbier}}, \bibinfo {author} {\bibfnamefont {A.}~\bibnamefont {Goussev}},\
  and\ \bibinfo {author} {\bibfnamefont {S.~C.~L.}\ \bibnamefont
  {Srivastava}},\ }\bibfield  {title} {\bibinfo {title} {Unbounded quantum
  backflow in two dimensions},\ }\href
  {https://doi.org/10.1103/PhysRevA.107.032204} {\bibfield  {journal} {\bibinfo
   {journal} {Phys. Rev. A}\ }\textbf {\bibinfo {volume} {107}},\ \bibinfo
  {pages} {032204} (\bibinfo {year} {2023})}\BibitemShut {NoStop}%
\bibitem [{\citenamefont {Berry}(2010)}]{Berry10}%
  \BibitemOpen
  \bibfield  {author} {\bibinfo {author} {\bibfnamefont {M.~V.}\ \bibnamefont
  {Berry}},\ }\bibfield  {title} {\bibinfo {title} {Quantum backflow, negative
  kinetic energy, and optical retro-propagation},\ }\href
  {https://doi.org/10.1088/1751-8113/43/41/415302} {\bibfield  {journal}
  {\bibinfo  {journal} {J. Phys. A: Math. Theor.}\ }\textbf {\bibinfo {volume}
  {43}},\ \bibinfo {pages} {415302} (\bibinfo {year} {2010})}\BibitemShut
  {NoStop}%
\bibitem [{\citenamefont {Eliezer}\ \emph {et~al.}(2020)\citenamefont
  {Eliezer}, \citenamefont {Zacharias},\ and\ \citenamefont
  {Bahabad}}]{Eliezer20}%
  \BibitemOpen
  \bibfield  {author} {\bibinfo {author} {\bibfnamefont {Y.}~\bibnamefont
  {Eliezer}}, \bibinfo {author} {\bibfnamefont {T.}~\bibnamefont {Zacharias}},\
  and\ \bibinfo {author} {\bibfnamefont {A.}~\bibnamefont {Bahabad}},\
  }\bibfield  {title} {\bibinfo {title} {Observation of optical backflow},\
  }\href {https://doi.org/10.1364/OPTICA.371494} {\bibfield  {journal}
  {\bibinfo  {journal} {Optica}\ }\textbf {\bibinfo {volume} {7}},\ \bibinfo
  {pages} {72} (\bibinfo {year} {2020})}\BibitemShut {NoStop}%
\bibitem [{\citenamefont {Daniel}\ \emph {et~al.}(2022)\citenamefont {Daniel},
  \citenamefont {Ghosh}, \citenamefont {Gorzkowski},\ and\ \citenamefont
  {Lapkiewicz}}]{Daniel22}%
  \BibitemOpen
  \bibfield  {author} {\bibinfo {author} {\bibfnamefont {A.}~\bibnamefont
  {Daniel}}, \bibinfo {author} {\bibfnamefont {B.}~\bibnamefont {Ghosh}},
  \bibinfo {author} {\bibfnamefont {B.}~\bibnamefont {Gorzkowski}},\ and\
  \bibinfo {author} {\bibfnamefont {R.}~\bibnamefont {Lapkiewicz}},\ }\bibfield
   {title} {\bibinfo {title} {Demonstrating backflow in classical two beams’
  interference},\ }\href {https://doi.org/10.1088/1367-2630/aca70b} {\bibfield
  {journal} {\bibinfo  {journal} {New J. Phys.}\ }\textbf {\bibinfo {volume}
  {24}},\ \bibinfo {pages} {123011} (\bibinfo {year} {2022})}\BibitemShut
  {NoStop}%
\bibitem [{\citenamefont {Palmero}\ \emph {et~al.}(2013)\citenamefont
  {Palmero}, \citenamefont {Torrontegui}, \citenamefont {Muga},\ and\
  \citenamefont {Modugno}}]{Palmero13}%
  \BibitemOpen
  \bibfield  {author} {\bibinfo {author} {\bibfnamefont {M.}~\bibnamefont
  {Palmero}}, \bibinfo {author} {\bibfnamefont {E.}~\bibnamefont
  {Torrontegui}}, \bibinfo {author} {\bibfnamefont {J.~G.}\ \bibnamefont
  {Muga}},\ and\ \bibinfo {author} {\bibfnamefont {M.}~\bibnamefont
  {Modugno}},\ }\bibfield  {title} {\bibinfo {title} {Detecting quantum
  backflow by the density of a {Bose-Einstein} condensate},\ }\href
  {https://doi.org/10.1103/PhysRevA.87.053618} {\bibfield  {journal} {\bibinfo
  {journal} {Phys. Rev. A}\ }\textbf {\bibinfo {volume} {87}},\ \bibinfo
  {pages} {053618} (\bibinfo {year} {2013})}\BibitemShut {NoStop}%
\bibitem [{\citenamefont {Mardonov}\ \emph {et~al.}(2014)\citenamefont
  {Mardonov}, \citenamefont {Palmero}, \citenamefont {Modugno}, \citenamefont
  {Sherman},\ and\ \citenamefont {Muga}}]{Mardonov14}%
  \BibitemOpen
  \bibfield  {author} {\bibinfo {author} {\bibfnamefont {S.}~\bibnamefont
  {Mardonov}}, \bibinfo {author} {\bibfnamefont {M.}~\bibnamefont {Palmero}},
  \bibinfo {author} {\bibfnamefont {M.}~\bibnamefont {Modugno}}, \bibinfo
  {author} {\bibfnamefont {E.~Y.}\ \bibnamefont {Sherman}},\ and\ \bibinfo
  {author} {\bibfnamefont {J.~G.}\ \bibnamefont {Muga}},\ }\bibfield  {title}
  {\bibinfo {title} {Interference of spin-orbit–coupled {Bose-Einstein}
  condensates},\ }\href {https://doi.org/10.1209/0295-5075/106/60004}
  {\bibfield  {journal} {\bibinfo  {journal} {EPL (Europhysics Lett.)}\
  }\textbf {\bibinfo {volume} {106}},\ \bibinfo {pages} {60004} (\bibinfo
  {year} {2014})}\BibitemShut {NoStop}%
\bibitem [{\citenamefont {Miller}\ \emph {et~al.}(2021)\citenamefont {Miller},
  \citenamefont {Yuan}, \citenamefont {Dumke},\ and\ \citenamefont
  {Paterek}}]{Miller21}%
  \BibitemOpen
  \bibfield  {author} {\bibinfo {author} {\bibfnamefont {M.}~\bibnamefont
  {Miller}}, \bibinfo {author} {\bibfnamefont {W.~C.}\ \bibnamefont {Yuan}},
  \bibinfo {author} {\bibfnamefont {R.}~\bibnamefont {Dumke}},\ and\ \bibinfo
  {author} {\bibfnamefont {T.}~\bibnamefont {Paterek}},\ }\bibfield  {title}
  {\bibinfo {title} {Experiment-friendly formulation of quantum backflow},\
  }\href {https://doi.org/10.22331/q-2021-01-11-379} {\bibfield  {journal}
  {\bibinfo  {journal} {{Quantum}}\ }\textbf {\bibinfo {volume} {5}},\ \bibinfo
  {pages} {379} (\bibinfo {year} {2021})}\BibitemShut {NoStop}%
\bibitem [{\citenamefont {Barbier}\ and\ \citenamefont
  {Goussev}(2021)}]{Barbier21}%
  \BibitemOpen
  \bibfield  {author} {\bibinfo {author} {\bibfnamefont {M.}~\bibnamefont
  {Barbier}}\ and\ \bibinfo {author} {\bibfnamefont {A.}~\bibnamefont
  {Goussev}},\ }\bibfield  {title} {\bibinfo {title} {On the
  experiment-friendly formulation of quantum backflow},\ }\href
  {https://doi.org/10.22331/q-2021-09-07-536} {\bibfield  {journal} {\bibinfo
  {journal} {{Quantum}}\ }\textbf {\bibinfo {volume} {5}},\ \bibinfo {pages}
  {536} (\bibinfo {year} {2021})}\BibitemShut {NoStop}%
\bibitem [{\citenamefont {Lundeen}\ \emph {et~al.}(2011)\citenamefont
  {Lundeen}, \citenamefont {Sutherland}, \citenamefont {Patel}, \citenamefont
  {Stewart},\ and\ \citenamefont {Bamber}}]{Lundeen11}%
  \BibitemOpen
  \bibfield  {author} {\bibinfo {author} {\bibfnamefont {J.~S.}\ \bibnamefont
  {Lundeen}}, \bibinfo {author} {\bibfnamefont {B.}~\bibnamefont {Sutherland}},
  \bibinfo {author} {\bibfnamefont {A.}~\bibnamefont {Patel}}, \bibinfo
  {author} {\bibfnamefont {C.}~\bibnamefont {Stewart}},\ and\ \bibinfo {author}
  {\bibfnamefont {C.}~\bibnamefont {Bamber}},\ }\bibfield  {title} {\bibinfo
  {title} {Direct measurement of the quantum wavefunction},\ }\href
  {https://doi.org/https://doi.org/10.1038/nature10120} {\bibfield  {journal}
  {\bibinfo  {journal} {Nature}\ }\textbf {\bibinfo {volume} {474}},\ \bibinfo
  {pages} {188} (\bibinfo {year} {2011})}\BibitemShut {NoStop}%
\bibitem [{\citenamefont {Pan}\ \emph {et~al.}(2019)\citenamefont {Pan},
  \citenamefont {Xu}, \citenamefont {Kedem}, \citenamefont {Wang},
  \citenamefont {Chen}, \citenamefont {Jan}, \citenamefont {Sun}, \citenamefont
  {Xu}, \citenamefont {Han}, \citenamefont {Li},\ and\ \citenamefont
  {Guo}}]{Pan19}%
  \BibitemOpen
  \bibfield  {author} {\bibinfo {author} {\bibfnamefont {W.-W.}\ \bibnamefont
  {Pan}}, \bibinfo {author} {\bibfnamefont {X.-Y.}\ \bibnamefont {Xu}},
  \bibinfo {author} {\bibfnamefont {Y.}~\bibnamefont {Kedem}}, \bibinfo
  {author} {\bibfnamefont {Q.-Q.}\ \bibnamefont {Wang}}, \bibinfo {author}
  {\bibfnamefont {Z.}~\bibnamefont {Chen}}, \bibinfo {author} {\bibfnamefont
  {M.}~\bibnamefont {Jan}}, \bibinfo {author} {\bibfnamefont {K.}~\bibnamefont
  {Sun}}, \bibinfo {author} {\bibfnamefont {J.-S.}\ \bibnamefont {Xu}},
  \bibinfo {author} {\bibfnamefont {Y.-J.}\ \bibnamefont {Han}}, \bibinfo
  {author} {\bibfnamefont {C.-F.}\ \bibnamefont {Li}},\ and\ \bibinfo {author}
  {\bibfnamefont {G.-C.}\ \bibnamefont {Guo}},\ }\bibfield  {title} {\bibinfo
  {title} {Direct measurement of a nonlocal entangled quantum state},\ }\href
  {https://doi.org/10.1103/PhysRevLett.123.150402} {\bibfield  {journal}
  {\bibinfo  {journal} {Phys. Rev. Lett.}\ }\textbf {\bibinfo {volume} {123}},\
  \bibinfo {pages} {150402} (\bibinfo {year} {2019})}\BibitemShut {NoStop}%
\bibitem [{\citenamefont {Zhang}\ \emph {et~al.}(2019)\citenamefont {Zhang},
  \citenamefont {Zhou}, \citenamefont {Mei}, \citenamefont {Liao},
  \citenamefont {Wen}, \citenamefont {Li}, \citenamefont {Zhang}, \citenamefont
  {Du}, \citenamefont {Yan},\ and\ \citenamefont {Zhu}}]{Zhang19}%
  \BibitemOpen
  \bibfield  {author} {\bibinfo {author} {\bibfnamefont {S.}~\bibnamefont
  {Zhang}}, \bibinfo {author} {\bibfnamefont {Y.}~\bibnamefont {Zhou}},
  \bibinfo {author} {\bibfnamefont {Y.}~\bibnamefont {Mei}}, \bibinfo {author}
  {\bibfnamefont {K.}~\bibnamefont {Liao}}, \bibinfo {author} {\bibfnamefont
  {Y.-L.}\ \bibnamefont {Wen}}, \bibinfo {author} {\bibfnamefont
  {J.}~\bibnamefont {Li}}, \bibinfo {author} {\bibfnamefont {X.-D.}\
  \bibnamefont {Zhang}}, \bibinfo {author} {\bibfnamefont {S.}~\bibnamefont
  {Du}}, \bibinfo {author} {\bibfnamefont {H.}~\bibnamefont {Yan}},\ and\
  \bibinfo {author} {\bibfnamefont {S.-L.}\ \bibnamefont {Zhu}},\ }\bibfield
  {title} {\bibinfo {title} {$\ensuremath{\delta}$-quench measurement of a pure
  quantum-state wave function},\ }\href
  {https://doi.org/10.1103/PhysRevLett.123.190402} {\bibfield  {journal}
  {\bibinfo  {journal} {Phys. Rev. Lett.}\ }\textbf {\bibinfo {volume} {123}},\
  \bibinfo {pages} {190402} (\bibinfo {year} {2019})}\BibitemShut {NoStop}%
\bibitem [{\citenamefont {Sahoo}\ \emph {et~al.}(2020)\citenamefont {Sahoo},
  \citenamefont {Chakraborti}, \citenamefont {Pati},\ and\ \citenamefont
  {Sinha}}]{Sahoo20}%
  \BibitemOpen
  \bibfield  {author} {\bibinfo {author} {\bibfnamefont {S.~N.}\ \bibnamefont
  {Sahoo}}, \bibinfo {author} {\bibfnamefont {S.}~\bibnamefont {Chakraborti}},
  \bibinfo {author} {\bibfnamefont {A.~K.}\ \bibnamefont {Pati}},\ and\
  \bibinfo {author} {\bibfnamefont {U.}~\bibnamefont {Sinha}},\ }\bibfield
  {title} {\bibinfo {title} {Quantum state interferography},\ }\href
  {https://doi.org/10.1103/PhysRevLett.125.123601} {\bibfield  {journal}
  {\bibinfo  {journal} {Phys. Rev. Lett.}\ }\textbf {\bibinfo {volume} {125}},\
  \bibinfo {pages} {123601} (\bibinfo {year} {2020})}\BibitemShut {NoStop}%
\bibitem [{\citenamefont {Fabre}\ \emph {et~al.}(2011)\citenamefont {Fabre},
  \citenamefont {Cheiney}, \citenamefont {Gattobigio}, \citenamefont
  {Vermersch}, \citenamefont {Faure}, \citenamefont {Mathevet}, \citenamefont
  {Lahaye},\ and\ \citenamefont {Gu\'ery-Odelin}}]{Fabre11}%
  \BibitemOpen
  \bibfield  {author} {\bibinfo {author} {\bibfnamefont {C.~M.}\ \bibnamefont
  {Fabre}}, \bibinfo {author} {\bibfnamefont {P.}~\bibnamefont {Cheiney}},
  \bibinfo {author} {\bibfnamefont {G.~L.}\ \bibnamefont {Gattobigio}},
  \bibinfo {author} {\bibfnamefont {F.}~\bibnamefont {Vermersch}}, \bibinfo
  {author} {\bibfnamefont {S.}~\bibnamefont {Faure}}, \bibinfo {author}
  {\bibfnamefont {R.}~\bibnamefont {Mathevet}}, \bibinfo {author}
  {\bibfnamefont {T.}~\bibnamefont {Lahaye}},\ and\ \bibinfo {author}
  {\bibfnamefont {D.}~\bibnamefont {Gu\'ery-Odelin}},\ }\bibfield  {title}
  {\bibinfo {title} {Realization of a distributed {B}ragg reflector for
  propagating guided matter waves},\ }\href
  {https://doi.org/10.1103/PhysRevLett.107.230401} {\bibfield  {journal}
  {\bibinfo  {journal} {Phys. Rev. Lett.}\ }\textbf {\bibinfo {volume} {107}},\
  \bibinfo {pages} {230401} (\bibinfo {year} {2011})}\BibitemShut {NoStop}%
\bibitem [{\citenamefont {Jendrzejewski}\ \emph {et~al.}(2012)\citenamefont
  {Jendrzejewski}, \citenamefont {M\"uller}, \citenamefont {Richard},
  \citenamefont {Date}, \citenamefont {Plisson}, \citenamefont {Bouyer},
  \citenamefont {Aspect},\ and\ \citenamefont {Josse}}]{Jendrzejewski12}%
  \BibitemOpen
  \bibfield  {author} {\bibinfo {author} {\bibfnamefont {F.}~\bibnamefont
  {Jendrzejewski}}, \bibinfo {author} {\bibfnamefont {K.}~\bibnamefont
  {M\"uller}}, \bibinfo {author} {\bibfnamefont {J.}~\bibnamefont {Richard}},
  \bibinfo {author} {\bibfnamefont {A.}~\bibnamefont {Date}}, \bibinfo {author}
  {\bibfnamefont {T.}~\bibnamefont {Plisson}}, \bibinfo {author} {\bibfnamefont
  {P.}~\bibnamefont {Bouyer}}, \bibinfo {author} {\bibfnamefont
  {A.}~\bibnamefont {Aspect}},\ and\ \bibinfo {author} {\bibfnamefont
  {V.}~\bibnamefont {Josse}},\ }\bibfield  {title} {\bibinfo {title} {Coherent
  backscattering of ultracold atoms},\ }\href
  {https://doi.org/10.1103/PhysRevLett.109.195302} {\bibfield  {journal}
  {\bibinfo  {journal} {Phys. Rev. Lett.}\ }\textbf {\bibinfo {volume} {109}},\
  \bibinfo {pages} {195302} (\bibinfo {year} {2012})}\BibitemShut {NoStop}%
\bibitem [{\citenamefont {Dupont}\ \emph {et~al.}(2023)\citenamefont {Dupont},
  \citenamefont {Gabardos}, \citenamefont {Arrouas}, \citenamefont {Ombredane},
  \citenamefont {Billy}, \citenamefont {Peaudecerf},\ and\ \citenamefont
  {Gu\'ery-Odelin}}]{Dupont23}%
  \BibitemOpen
  \bibfield  {author} {\bibinfo {author} {\bibfnamefont {N.}~\bibnamefont
  {Dupont}}, \bibinfo {author} {\bibfnamefont {L.}~\bibnamefont {Gabardos}},
  \bibinfo {author} {\bibfnamefont {F.}~\bibnamefont {Arrouas}}, \bibinfo
  {author} {\bibfnamefont {N.}~\bibnamefont {Ombredane}}, \bibinfo {author}
  {\bibfnamefont {J.}~\bibnamefont {Billy}}, \bibinfo {author} {\bibfnamefont
  {B.}~\bibnamefont {Peaudecerf}},\ and\ \bibinfo {author} {\bibfnamefont
  {D.}~\bibnamefont {Gu\'ery-Odelin}},\ }\bibfield  {title} {\bibinfo {title}
  {Hamiltonian ratchet for matter-wave transport},\ }\href
  {https://doi.org/10.1103/PhysRevLett.131.133401} {\bibfield  {journal}
  {\bibinfo  {journal} {Phys. Rev. Lett.}\ }\textbf {\bibinfo {volume} {131}},\
  \bibinfo {pages} {133401} (\bibinfo {year} {2023})}\BibitemShut {NoStop}%
\bibitem [{\citenamefont {Galapon}\ \emph {et~al.}(2005)\citenamefont
  {Galapon}, \citenamefont {Delgado}, \citenamefont {Muga},\ and\ \citenamefont
  {Egusquiza}}]{Galapon05}%
  \BibitemOpen
  \bibfield  {author} {\bibinfo {author} {\bibfnamefont {E.~A.}\ \bibnamefont
  {Galapon}}, \bibinfo {author} {\bibfnamefont {F.}~\bibnamefont {Delgado}},
  \bibinfo {author} {\bibfnamefont {J.~G.}\ \bibnamefont {Muga}},\ and\
  \bibinfo {author} {\bibfnamefont {I.~n.}\ \bibnamefont {Egusquiza}},\
  }\bibfield  {title} {\bibinfo {title} {Transition from discrete to continuous
  time-of-arrival distribution for a quantum particle},\ }\href
  {https://doi.org/10.1103/PhysRevA.72.042107} {\bibfield  {journal} {\bibinfo
  {journal} {Phys. Rev. A}\ }\textbf {\bibinfo {volume} {72}},\ \bibinfo
  {pages} {042107} (\bibinfo {year} {2005})}\BibitemShut {NoStop}%
\bibitem [{\citenamefont {Hegerfeldt}\ \emph {et~al.}(2003)\citenamefont
  {Hegerfeldt}, \citenamefont {Seidel},\ and\ \citenamefont
  {Gonzalo~Muga}}]{Hegerfeldt03}%
  \BibitemOpen
  \bibfield  {author} {\bibinfo {author} {\bibfnamefont {G.~C.}\ \bibnamefont
  {Hegerfeldt}}, \bibinfo {author} {\bibfnamefont {D.}~\bibnamefont {Seidel}},\
  and\ \bibinfo {author} {\bibfnamefont {J.}~\bibnamefont {Gonzalo~Muga}},\
  }\bibfield  {title} {\bibinfo {title} {Quantum arrival times and operator
  normalization},\ }\href {https://doi.org/10.1103/PhysRevA.68.022111}
  {\bibfield  {journal} {\bibinfo  {journal} {Phys. Rev. A}\ }\textbf {\bibinfo
  {volume} {68}},\ \bibinfo {pages} {022111} (\bibinfo {year}
  {2003})}\BibitemShut {NoStop}%
\bibitem [{\citenamefont {Sombillo}\ and\ \citenamefont
  {Galapon}(2016)}]{Galapon16}%
  \BibitemOpen
  \bibfield  {author} {\bibinfo {author} {\bibfnamefont {D.~L.~B.}\
  \bibnamefont {Sombillo}}\ and\ \bibinfo {author} {\bibfnamefont {E.~A.}\
  \bibnamefont {Galapon}},\ }\bibfield  {title} {\bibinfo {title} {Particle
  detection and non-detection in a quantum time of arrival measurement},\
  }\href {https://doi.org/https://doi.org/10.1016/j.aop.2015.11.008} {\bibfield
   {journal} {\bibinfo  {journal} {Annals of Physics}\ }\textbf {\bibinfo
  {volume} {364}},\ \bibinfo {pages} {261} (\bibinfo {year}
  {2016})}\BibitemShut {NoStop}%
\end{thebibliography}%

\end{document}